\documentclass[sigconf]{acmart}
\setcopyright{rightsretained}
\makeatletter

\acmConference[ICCPS '23]{14TH ACM/IEEE International Conference On Cyber-Physical Systems}{May 09--12,
  2023}{San Antonio, TX}

\pdfoutput=1 

\usepackage{graphicx}
\usepackage{textcomp}
\usepackage{xcolor, colortbl}
\usepackage{multirow}
\usepackage{subfigure}
\usepackage{amsmath}
\usepackage{amsfonts}
\usepackage{amsxtra}
\usepackage{amsthm}
\usepackage[mathscr]{eucal}

\usepackage{enumitem}
\usepackage{lipsum}

\usepackage{float}

\usepackage{calc}

\usepackage{cancel}
\usepackage{array}

\usepackage{amsbsy}
\usepackage[linesnumbered,lined,boxed,commentsnumbered]{algorithm2e}

\SetCommentSty{mycommfont}

\makeatletter
\newcommand\Func[2]{%
    \textbf{function} #1
    \algocf@group{#2}%
}
\makeatother

\makeatletter
\newcommand\Forr[2]{%
    \textbf{for} #1 \textbf{do}%
    \algocf@group{#2}%
}
\makeatother

\makeatletter
\newcommand\Blnk[2]{%
    \hspace{10pt}#1
    \algocf@group{#2}%
    \textbf{end}
}
\makeatother

\newtheorem{theorem}{Theorem}
\newtheorem{definition}{Definition}

\newtheorem{corollary}{Corollary}
\newtheorem{proposition}{Proposition}
\newtheorem{lemma}{Lemma}

\newtheorem{assumption}{Assumption}
\newtheorem{problem}{Problem}
\newtheorem*{problem*}{Problem}

\newcommand{\nn}{{\mathscr{N}\mathscr{N}}}

\newcolumntype{C}[1]{>{\centering$}p{#1}<{$}}
\newcommand{\timeIdx}[1]{$\begin{array}{C{5pt} @{} C{30pt} @{} C{5pt}} [ & #1 & ] \end{array}$}

\usepackage{soul}

\usepackage{color} 

\newcommand{\myalg}{{E}nergy{S}hield}

\definecolor{darkgreen}{RGB}{0,120,0}

\AtBeginDocument{%
  \providecommand\BibTeX{{%
    \normalfont B\kern-0.5em{\scshape i\kern-0.25em b}\kern-0.8em\TeX}}}

\begin{document}

\title{EnergyShield: Provably-Safe Offloading of Neural Network Controllers for Energy Efficiency}


\title{EnergyShield: Energy Conservation by off-device Evaluation of Neural Network Controllers with Safety Guarantees}

\title{EnergyShield: Provably-Safe Offloading of Neural Network Controllers for Energy Efficiency}
%


\author{Mohanad Odema, James Ferlez, Goli Vaisi, Yasser Shoukry, Mohammad Abdullah Al Faruque}
\affiliation{%
  \institution{Department of Electrical Engineering and Computer Science} \country{University of California, Irvine, CA, USA}
}









\begin{abstract}
To mitigate the high energy demand of Neural Network (NN) based Autonomous Driving 
Systems (ADSs), we consider the problem of offloading NN 
controllers from the ADS to nearby edge-computing infrastructure, but in such a 
way that formal \emph{vehicle} safety properties are guaranteed. In particular, we propose the 
\myalg~framework, which repurposes a controller ``shield'' as a low-power 
runtime safety monitor for the ADS vehicle. Specifically, the shield in 
\myalg~provides not only safety interventions but also a formal, state-based 
quantification of the tolerable edge response time before vehicle safety is  
compromised. Using \myalg, an ADS can then save energy by wirelessly offloading 
NN computations to edge computers, while still maintaining a formal guarantee 
of safety until it receives a response (on-vehicle hardware provides a 
just-in-time fail safe). To validate the benefits of \myalg, we implemented and 
tested it in the Carla simulation environment. Our results show that 
\myalg~maintains safe vehicle operation while providing significant energy 
savings compared to on-vehicle NN evaluation: from 24\% to 54\% less energy 
across a range of wireless conditions and edge delays.

\end{abstract}






\maketitle

%

\section{Introduction} 
\label{sec:introduction}
\let\thefootnote\relax\footnote{This work was partially supported by the National Science Foundation (NSF) under awards CCF-2140154, CNS-2002405, and ECCS-2139781 and by the C3.ai Digital Transformation Institute.}
Advances in the theory and application of Neural Networks (NNs), particularly 
Deep NNs (DNNs), have spurred revolutionary progress on a number of AI tasks, 
including perception, motion planning and control. As as result, DNNs have 
provided a feasible engineering solution to supplant formerly human-only tasks, 
most ambitiously in Autonomous Driving Systems (ADSs). However, 
state-of-the-art ADSs require the use of very large DNN architectures to solve  
essential perception and control tasks, which generally involve  processing the 
output of tens of cameras, LiDARs and other sensors. As a result, contemporary 
ADSs are only possible with significant computational resources deployed on the 
vehicle itself, since their DNNs must process such high-bandwidth sensors in 
closed-loop, in real time. The practical energy impact of high-capacity 
on-vehicle compute is understudied, but current research suggests that it is 
profound: e.g., up to a 15\% reduction in a vehicle's range 
\cite{lin2018architectural, mohan2020trade}.


At the same time, advances in semiconductor design and packaging have made 
possible cheap, low-power silicon; and advances in wireless networking have 
made high-bandwidth, low-latency radio links possible even in challenging 
multi-user environments. Together, these advances have led to increasingly 
ubiquitous, cheap, wirelessly accessible computational resources near the 
\emph{edge} of conventional hard-wired infrastructure. In particular, it is now 
possible to achieve reliable, millisecond-latency wireless connections between 
connected ADSs and nearby edge computing \cite{liu2019edge, malawade2021sage, 
baidya2020vehicular}.  

The ubiquity of edge compute thus suggests a natural way to reduce the 
\emph{local} energy consumption on ADS vehicles: viz., by wirelessly 
\emph{offloading} onerous perception and control DNN computations to abundant 
nearby edge compute infrastructure. However, even modern wireless networks and 
offloading-friendly  DNN architectures (e.g. encoder/decoders) cannot provide 
\emph{formal guarantees} that bringing edge computing ``into the loop'' will 
have equivalent (or even acceptable) performance compared to on-vehicle 
hardware. This is an  unacceptable situation when human lives are at stake: 
even relatively rare and short delays in obtaining a control action or 
perception classification can have fatal consequences.

\begin{figure}[!t] %
	\centering 
	\includegraphics[,width = 0.49\textwidth]{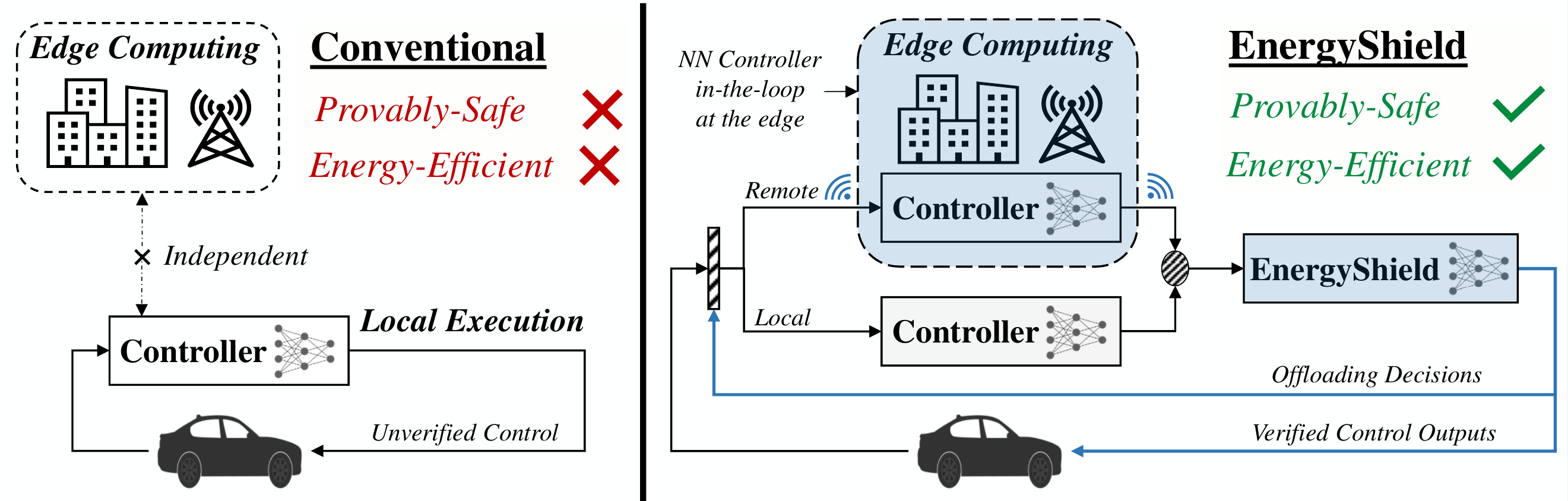} 
        \vspace{-4ex}
	\caption{Illustration of Provably-Safe Offloading of Neural Network Controllers for Energy Efficiency. 
	} 
	\label{fig:energy_shield_generic}
 \vspace{-2ex}
\end{figure}  %

In this paper, we propose the \myalg~framework as a mechanism to perform 
DNN-to-edge offloading of ADS controllers but \emph{in a formal, provably safe 
way}. Thus, \myalg~is, to the best of our knowledge, the first framework that 
enables significantly lower on-vehicle energy usage when evaluating large DNNs 
by intelligently offloading those calculations to edge compute in a provably 
safe way; see Figure \ref{fig:energy_shield_generic}. The primary idea of 
\myalg~is to perform safety-aware (state-)contextual offloading of DNN 
calculations to the edge, under the assumption that adequate on-vehicle 
computation is always available as a safety fallback. This is accomplished 
using a controller ``shield'' as both a mechanism to enforce safety \emph{and} 
as a novel, low-power runtime safety monitor. In particular, this shield-based 
safety monitor provides provably safe edge-compute response times: i.e., at 
each instance, \myalg~provides a time interval within which the edge-compute 
must respond to an offloading request in order to \emph{guarantee safety of the 
vehicle} in the interim. In the event that the edge resources don't provide a 
response within this time, on-board local compute proceeds to evaluate the 
relevant DNNs before vehicle safety is no longer assured. Further energy 
savings are obtained by incorporating an estimator to anticipate edge-compute 
load and wireless network throughput; a more intelligent offloading decision is 
made by comparing this estimate against the tolerable edge-compute delay 
provided by the runtime safety monitor -- thus preventing offloads that are 
likely to  fail. 

The main technical novelty of \myalg~is its shield-based runtime safety monitor 
mentioned above. Although controller ``shielding'' is a well-known methodology 
to render generic controllers safe, the shielding aspect of \myalg~contains two 
important novel contributions of its own: first, in the use of a shield not 
only to enforce safety but also as a runtime safety monitor to quantify the 
\emph{time until the system is unsafe}; and second, in the specific design of 
that runtime monitor with regard to implementation complexity and energy 
considerations. In the first case, \myalg~extends existing notions wherein the 
\emph{current value} of a (Zeroing-)Barrier Function (ZBF) is used as a runtime 
monitor to quantify the safety of an agent: in particular, it is novel in 
\myalg~to instead use the current value of the ZBF to derive a \emph{sound} 
quantification of the \emph{time} until the agent becomes unsafe. Moreover, 
\myalg~implements this quantification in an extremely energy efficient way: 
i.e., via a small lookup table that requires only a small number of FLOPS to 
obtain a guaranteed time-until-unsafe. This particular aspect of the runtime 
safety monitor is also facilitated by a using a particular, but known, ZBF and 
shield \cite{ferlez2020shieldnn} in \myalg: these components are both extremely 
simple, and so implementable using small, energy efficient NNs 
\cite{ferlez2020shieldnn}. Together, these design choices ensures that any 
energy saved by offloading is not subsequently expended in the implementation 
of \myalg~itself.

We conclude the paper with a significant sequence of experiments to validate 
both the safety and energy savings provided by the \myalg~framework. In 
particular, we tested \myalg~in the Carla simulation environment 
\cite{dosovitskiy2017carla} with several Reinforcement Learning (RL)-trained 
agents. Our experiments showed that \myalg~entirely eliminated obstacle 
collisions for the RL agents we considered -- i.e. made them safe -- while 
simultaneously reducing NN energy consumption by as much as 54\%. Additionally, 
we showed that \myalg~has intuitive, safety-conscious offloading behavior: when 
the ADS is near an obstacle -- and hence less safe -- \myalg's runtime safety 
monitor effectively forced exclusively on-vehicle NN evaluation; when the ADS 
was further from an obstacle -- and hence more safe -- \myalg's runtime safety 
monitor allowed more offloading, and hence more energy savings.

\noindent \textbf{Related Work:~} \emph{Formal Methods for Data-Trained  
Controllers.} 
A number of approaches exist to assert the safety of data-trained controllers 
with formal guarantees; in most, ideas from control theory are used in some  
way to augment the trained controllers to this end. A good survey of these methods is \cite{DawsonSafeControlLearned2022}. Examples of this approach 
include the use of Lyapunov 
methods~\cite{berkenkamp2017safe,chow2019lyapunov,chow2018lyapunov}, safe model 
predictive control~\cite{koller2018learning}, reachability 
analysis~\cite{akametalu2014reachability,govindarajan2017data,fisac2018general}, 
barrier certificates~\cite{wang2018safe,
srinivasan2020synthesis,%
llanes2022safety,%
xiao2021rule,%
taylor2020control
,robey2020learning}, and online learning of  
uncertainties~\cite{shi2019neural}.
Controller ``shielding'' \cite{AlshiekhSafeReinforcementLearning2017} is  
another technique that often falls in the barrier function category 
\cite{ChengEndtoEndSafeReinforcement2019}. Another approach tries to verify the 
formal safety properties of learned controllers using formal verification 
techniques (e.g., model checking): e.g., through the use of SMT-like 
solvers~\cite{sun2019formal,dutta2018output,liu2019algorithms} or hybrid-system 
verification~\cite{fazlyab2019efficient,xiang2019reachable,ivanov2019verisig}. 
However, these techniques only assess the safety of a given controller rather 
than design or train a safe agent.



\emph{Edge Computing for Autonomous Systems.} A number of different edge/cloud 
offloading schemes have been proposed for ADSs, however none to date has  
provided formal guarantees. Some have focused on scheduling techniques and 
network topology to achieve effective offloading \cite{tang2021vecframe,%
zhang2021emp, feng2018mobile,cui2020offloading,sasaki2016vehicle}. %
Others focused on split and other NN architectures to make offloading more 
efficient \cite{malawade2021sage,odema2022testudo, chen2022romanus}.

\section{Preliminaries} 
\label{sec:preliminaries}

\subsection{Notation} 
\label{sub:notation}
Let $\mathbb{R}$ denote the real numbers; $\mathbb{R}^+$ the non-negative real 
numbers; $\mathbb{N}$ the natural numbers; and $\mathbb{Z}$ the integers. For a 
continuous-time signal, $x(t), t \geq 0$, denote its discrete-time sampled 
version as $x[n]$ for some fixed sample period $T$ (in seconds); i.e. let $x[n] 
\triangleq x(n\cdot T)$ for $n\in \mathbb{Z}$. Let  $\mathbf{1}_a : \mathbb{R} 
\rightarrow \{a\}$ be the constant function with value $a$; i.e., 
$\mathbf{1}_a(x) = a$ for all $x\in \mathbb{R}$ (interpreted as a sequence 
where necessary). Finally, let $\dot{x} = f(x , u)$ be a generic control system 
with $x \in \mathbb{R}^n$ and $u\in  \mathbb{R}^m$, and let $\pi : \mathbb{R} 
\times \mathbb{R}^n \rightarrow \mathbb{R}^m$ be a (possibly) time-varying 
controller. For this system and controller, consider a time $\tau \geq 0$ and 
state $x_0$, and denote by $\zeta_\pi^{t_0,x_0}: \mathbb{R}^+ \rightarrow 
\mathbb{R}^n$ the $t_0$-shifted state evolution of this system under controller 
$\pi$ assuming  $x(t_0) = x_0$. Let $\zeta_\pi^{n_0,x_0}[n]$ indicate the same, 
except in discrete-time and with zero-order hold application of $\pi$.  $\lVert 
\cdot \rVert$ and $\lVert \cdot \rVert_2$ will denote the $\max$ and two-norms 
on $\mathbb{R}^n$, respectively.

\subsection{Kinematic Bicycle Model} 
\label{sub:kinematic_bicycle_model}
\begin{figure} 
	\centering 
	\includegraphics[width=0.30\textwidth]{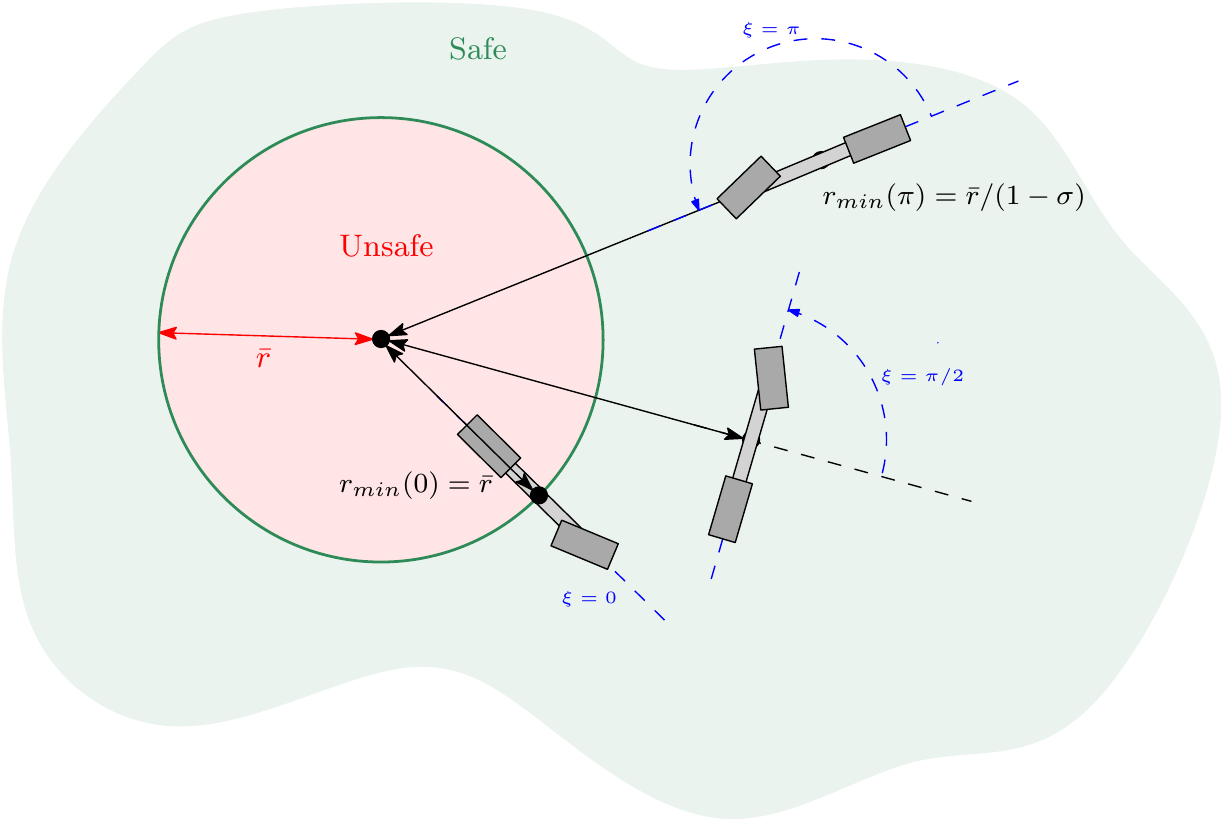} 
	\caption{Obstacle specification and minimum barrier distance as a function 
	of relative vehicle orientation, $\xi$. 
	} 
	\label{fig:obstacle_diagram}
\end{figure}  %
Throughout this paper, we will use the kinematic bicycle model (KBM) as the 
formal dynamical model for our autonomous vehicle 
\cite{KongKinematicDynamicVehicle2015}. However, we consider the KBM model in 
terms of state variables relative to a fixed point in the plane,  rather than 
absolute Cartesian coordinates. That is, the states are  the distance to a 
fixed point, $\lVert \vec{r} \rVert$, and orientation angle, $\xi$, of the 
vehicle with respect to the same. These quantities describe a state-space model:
\begin{equation}\label{eq:kbm_dynamics}
	\left(
		\begin{smallmatrix}
			\dot{r} \\
			\dot{\xi} \\
			\dot{v}
		\end{smallmatrix}
	\right)
	\negthickspace
	=
	\left(
		\begin{smallmatrix}
			v \cos( \xi - \beta ) \\
			-\frac{1}{r} v \sin(\xi - \beta) - \frac{v}{\ell_r} \sin(\beta) \\
			a 
		\end{smallmatrix}
	\right);
	\quad
	\beta \triangleq \tan^{-1}(\tfrac{\ell_r}{\ell_f + \ell_r} \tan(\delta_f))
\end{equation}
\noindent where $r(t)$ is distance to the origin; $a$ is the linear 
acceleration input; $\delta_f$ is the front-wheel steering angle 
input\footnote{That is the steering angle can be set instantaneously, and the 
dynamics of the steering rack can be ignored.}; and $\psi + \xi = 
\tan^{-1}(y/x)$. We note a few special states for convenience: at $\xi = \pm 
\pi/2$, the vehicle is oriented tangentially to the obstacle, and at $\xi = \pi 
\text{ or } 0$, the vehicle is pointing directly at or away from the origin, 
respectively (Figure \ref{fig:obstacle_diagram}). $\beta$ is an 
\emph{invertible function} of the steering control, $\delta_f$.

We also assume that the KBM has a steering constraint so that $\delta_f \in 
[-{\delta_f}_\text{max}, {\delta_f}_\text{max}]$. For simplicity, we will 
consider $\beta$ directly as a control variable; this is without loss of 
generality, for the reason noted above. Thus, $\beta$ is also constrained: 
$\beta \in [-\beta_\text{max}, \beta_\text{max}]$. 
%
Finally, we define the state and control vectors for the KBM as: $\chi 
\triangleq (\xi, r, v)$ and $\omega \triangleq (a , \beta)$, where $\omega \in 
\Omega_\text{admis.} \triangleq \mathbb{R} \times [-\beta_\text{max}, 
\beta_\text{max}]$, the set of admissible controls. Thus, the dynamics of the 
KBM are given by $\dot{\chi} = f_\text{KBM}(\chi, \omega)$ with $f_\text{KBM}$ 
defined according to \eqref{eq:kbm_dynamics}.


\subsection{Barrier Functions and Shielding} 
\label{sub:barrier_functions_and_shielding}

In the sequel, we will use a controller ``shield'', which is a methodology for 
instantaneously correcting the outputs produced by a controller in closed loop; 
the objective is to make corrections such that the original controller, however 
it was designed or implemented, becomes safe --  hence the ``shield'' moniker. 
Specifically, a controller shield is designed around a real-valued function 
over the state space of interest, called a (Zeroing-) Barrier Function (ZBF). 
The ZBF directly specifies a set safe states by its sign: states for which the 
ZBF is nonnegative are taken to be safe. The ZBF in turn indirectly specifies 
safe controls (as a function of state) in such a way that the sign of the ZBF 
is invariant along trajectories of the dynamics.

Formally, consider a control system $\dot{x} = f(x, u)$ in closed loop with a 
state-feedback controller $\pi: x \mapsto u$. In this scenario, a feedback 
controller in closed loop converts the control system into an autonomous one -- 
the autonomous vector field $f(\cdot , \pi( \cdot )) $. In the context of an  
autonomous system, then, we recall the following definition for a 
Zeroing-Barrier Function (ZBF):
\begin{definition}[Zeroing Barrier Function (ZBF) {\cite[Definition 2]{XuRobustnessControlBarrier2015}}]
	Let $\dot{x} = f(x, \pi(x))$ be the aforementioned closed-loop, autonomous 
	system with $x(t) \in \mathbb{R}^n$. Also, let $h : \mathbb{R}^n 
	\rightarrow \mathbb{R}$, and define $\mathcal{C} \triangleq \{x \in  
	\mathbb{R}^n : h(x) \geq 0\}$. If there exists a locally Lipschitz, 
	extended-class-K function,  $\alpha$ such that:
	\begin{equation}
		\nabla_x h(x)\cdot f(x, \pi(x)) \geq -\alpha(h(x)) \text{ for all } x\in \mathcal{C}
	\end{equation}
	then $h$ is said to be a \textbf{zeroing barrier function (ZBF)}.
\end{definition}
Moreover, the conditions for a barrier function above can be translated into a 
set membership problem for the outputs of such a feedback controller. This is 
explained in the following corollary.



\begin{proposition}
\label{prop:feedback_control_action_set}
	Let $\dot{x} = f(x,u)$ be a control system that is Lipschitz continuous in 
	both of its arguments on a set $\mathcal{D} \times \Omega_\text{admis.}$; 
	furthermore, let $h: \mathbb{R}^n \rightarrow \mathbb{R}$ with 
	$\mathcal{C}_h \triangleq \{x \in \mathbb{R}^n | h(x) \geq 0\} \subseteq 
	\mathcal{D}$, and let $\alpha$ be a class $\mathcal{K}$ function. If the set
	\begin{equation}\label{eq:safe_control_set}
			R_{h,\alpha}(x) \triangleq \{ u \in \Omega_\text{admis.} | \nabla_x^\text{T} h(x) \cdot f(x, u) + \alpha(h(x)) \geq 0 \}
	\end{equation}
	is non-empty for each $x \in \mathcal{D}$, and a Lipschitz  continuous 
	feedback controller $\pi : x \mapsto u$ satisfies
	\begin{equation}
	\label{eq:closed_loop_safe}
		\pi(x) \in R_{h,\alpha}(x) \quad \forall x \in \mathcal{D}
	\end{equation}
	then $\mathcal{C}_h$ is forward invariant for the closed-loop dynamics 
	$f(\cdot, \pi(\cdot))$.

	In particular, if $\pi$ satisfies \eqref{eq:closed_loop_safe} and $x(t)$  
	is a trajectory of $\dot{x} = f(x,\pi(x))$ with $h(x(0)) \geq 0$, then 
	$h(x(t)) \geq 0$ for all $t \geq 0$.
\end{proposition}
\begin{proof}
	This follows directly from an application of zeroing barrier functions 
	\cite[Theorem 1]{XuRobustnessControlBarrier2015}.
\end{proof}
Proposition \ref{prop:feedback_control_action_set} is the foundation for 
controller shielding: \eqref{eq:safe_control_set} and 
\eqref{eq:closed_loop_safe} establish that $h$ (and associated $\alpha$)  forms 
a ZBF for the closed-loop, autonomous dynamics $f(\cdot , \pi(\cdot))$ . Note 
also that there is no need to distinguish between a closed-loop feedback 
controller $\pi$, and a composite of $\pi$ with a function that \emph{shields} 
(or filters) its output based on the current state. Hence, the following 
definition:

\begin{definition}[Controller Shield]
\label{def:controller_shield}
	Let $\dot{x} = f(x, u)$, $h$, $\mathfrak{C}_h$, $\alpha$ and $\mathcal{D} 
	\times \Omega_\text{admis.}$ be as in Proposition 
	\ref{prop:feedback_control_action_set}. Then a \textbf{controller shield} 
	is a Lipschitz continuous function $\mathfrak{S} : \mathcal{D} \times 
	\Omega_\text{admis.} \rightarrow \Omega_\text{admis.}$ such that
	\begin{equation}
		\forall (x, u) \in \mathcal{D} \times \Omega_\text{admis.} . \mathfrak{S}(x, u) \in R_{h,\alpha}(x).
	\end{equation}
\end{definition}



\subsection{A Controller Shield for the KBM} 
\label{sub:kbm_controller_shield}
In this paper, we will make use of the existing ZBF and controller shield 
designed for the KBM in \cite{ferlez2020shieldnn}. These function in concert to 
provide controller shielding for the safety property also illustrated in Figure 
\ref{fig:obstacle_diagram}: i.e., to prevent the KBM from entering a disk of 
radius $\bar{r}$ centered at the origin.

In particular, \cite{ferlez2020shieldnn} proposes the following class of 
\emph{candidate} barrier functions for the KBM:
\begin{equation}
	\label{eq:bicycle_barrier}
	h_{\bar{r},\sigma}(\chi) = h_{\bar{r},\sigma}(\xi, r, v) = \tfrac{\sigma \cos(\xi/2) + 1-\sigma}{\bar{r}} - \tfrac{1}{r}
\end{equation}
where $\sigma \in (0,1)$ is the parameter that characterizes the class; this 
choice is combined with a class $\mathcal{K}$ function:
\begin{equation}
\label{eq:extended_class_k}
	\alpha_{v_\text{max}}(x) = K \cdot v_\text{max} \cdot x
\end{equation}
to form a parametric class of ZBFs. Note also that this class of ZBFs ignores 
the state variable, $v$; it is a result in \cite{ferlez2020shieldnn} that this 
class is useful as a barrier function provided the vehicle velocity remains  
(is controlled) within the range $(0, v_\text{max}]$. Note also that the 
equation has $h_{\bar{r},\sigma}(\chi) = 0$ has a convenient solution, which we 
denote by $r_\text{min}$ for future reference:
\begin{equation}
\label{eq:r_min}
	r_\text{min}(\xi) = \bar{r}/( \sigma \cos(\xi/2) + 1-\sigma ).
\end{equation}

One main result in \cite{ferlez2020shieldnn} is a mechanism for choosing the 
parameter $\sigma$ as a function of KBM parameters (e.g. $\ell_r$) and safety 
parameter, $\bar{r}$ so that the resulting specific function is indeed a ZBF as 
required. 


Finally, we note that \cite{ferlez2020shieldnn} also suggests an extremely 
lightweight implementation of barrier based on \eqref{eq:bicycle_barrier}. That 
is, it contains a ``Shield Synthesizer'' that implements a controller shield  
by approximating a simple single-input/single-output concave function with a 
ReLU NN \cite[pp 6]{ferlez2020shieldnn}. This construction will also prove 
advantageous later. We denote by $\mathfrak{S}_\text{KBM}$ the resulting 
controller shield, with associated barrier, KBM and safety parameters inferred 
from the context.

\section{Framework} 
\label{sec:framework}
\emph{NOTE: For the duration of this section, we will denote by $x$, $y$ and  
$u$ the state, sensor and control variables of an ADS, respectively; this 
abstract notation will illustrate the \myalg~framework free from  specific 
modelling details. A detailed consideration of \myalg, with formal proofs and 
assumptions, appears in the Section \ref{sec:shieldnn}.}

\subsection{\myalg~Motivation and Context} 
\label{sub:energyshield_motivation}
The basic motivation for the \myalg~framework is the following. Suppose that 
an ADS contains a large NN, $\nn_c$, that is responsible for producing a 
control action, $u$, in response to a sensor signal, $y$. Further assume that, 
by virtue of its size, computing an output of $\nn_c$ with on-vehicle hardware 
consumes significant energy. Thus, it would be advantageous, energy-wise, to 
\emph{offload} evaluations of $\nn_c$ to edge computing infrastructure: in 
other words,  wirelessly transmit a particular sensor measurement, $y$, to 
off-vehicle edge computers, where the output $u = \nn_c(y)$ is computed and 
returned to the requesting ADS.

The problem with this approach is largely from a safety point of view. In 
particular, the controller $\nn_c$ was designed to operate \emph{in real time  
and in closed-loop}: i.e. the control action at discrete-time sample $n$ is 
intended to be computed from the sensor measurement at the \emph{same} time 
sample\footnote{In our formal consideration, we will model a one-sample 
computation delay.}. In the notation of discrete-time signals (see Section  
\ref{sub:notation}), this means: $u[n] = \nn_c(y[n])$. However, offloading a 
sensor measurement, $y[n]$, to the edge entails that the correct output of 
$\nn_c(y[n])$ will not be back on-vehicle before some non-zero number of  
samples, say $\Delta$. Thus, $\nn_c(y[n])$ will not be available at time $n$ to 
set $u[n] = \nn_c(y[n])$ as intended; rather, the soonest possible use of the 
output $\nn_c(y[n])$ will be at time $n+\Delta$, or $u[n +  \Delta] = 
\nn_c(y[n])$. This delay creates obvious safety issues, since the state of the 
vehicle --  and hence the correct control to apply -- will have 
changed in the  intervening $\Delta$ time samples. More importantly, even the 
``correct'' control action applied at $n + \Delta$ may be insufficient to 
ensure safety: e.g., after $\Delta$ samples have elapsed, it may be too late to 
apply adequate evasive steering.


\subsection{\myalg~Structure} 
\label{sub:energyshield_structure}
If we assume that the ADS itself has enough on-vehicle compute to obtain an 
output $\nn_c(y[n])$ \emph{in real time}, then this safety problem above  
becomes one of making an \emph{offloading decision}: ideally one that saves  
energy without compromising safety. That is, should a particular computation of 
$\nn_c$ be offloaded to the edge? And how long should the ADS wait for a 
response so as to ensure the situation is correctable?

The nature of the offloading decision means that \myalg~must address two 
essential and intertwined issues in order to ensure safety of the ADS vehicle 
during offloading. On the one hand, \myalg~must be able to \emph{correct} the 
control actions provided by $\nn_c$ after an offload decision (see explanation 
above). On the other hand, \myalg~must limit the duration it waits for each 
offloading request, so that actions provided by $\nn_c$ \emph{can} be corrected 
in the first place; i.e., among all possible offloading delays, $\Delta$, it is 
not immediate which may be corrected and which may not (loosely speaking, 
$\Delta = \infty$ is one delay that likely cannot be  corrected). In this 
sense, knowing that a particular response-delay, $\Delta^\prime$ is correctable 
essentially characterizes how to take an offloading decision, since it provides 
an  \textbf{\itshape expiration on safety}: i.e., proceed to offload, and wait 
for a response until $\Delta^\prime$ samples have elapsed -- at which point 
resume local evaluation of $\nn_c$.

In particular, 
\myalg~has two central components to achieve these ends:
\begin{enumerate}[label={\bfseries C\arabic*:}]
	\item \textbf{Controller Shield.} \myalg~contains a controller shield  
		(see Section \ref{sub:barrier_functions_and_shielding}), which ensures 
		that safety is maintained irrespective of offloading-delayed controller 
		outputs; in other words, it corrects unsafe behavior of $\nn_c$ that 
		results from unanticipated (to $\nn_c$) changes in vehicle state during 
		offloading delays.

	\item \textbf{Runtime Safety Monitor.} \myalg~contains a runtime safety 
		monitor that provides the ADS an upper bound, $\Delta_\text{max}$ (in  
		samples), on how long it should wait for a response to one of its 
		offloading requests to maintain safety, \emph{assuming no updates to  
		the control action in the meantime}; i.e., provided the offload delay 
		is $\Delta \leq \Delta_\text{max}$, then \textbf{C1}, the controller 
		shield, can guarantee safe recovery after holding the last control 
		signal update through offload delay period (\textbf{C1} may use 
		on-vehicle computation if necessary). \textbf{\itshape In other  words, 
		}$\Delta_\text{max}$ \textbf{\itshape specifies and expiration  for the 
		safety guarantee provided by} \textbf{C1} \textbf{\itshape  using 
		on-vehicle computation.}
\end{enumerate}
\noindent  Naturally, \textbf{C1} and \textbf{C2} need to be designed together, 
since their objectives are mutually informed. Indeed, in the specific design of 
\myalg, these components are designed from the same ZBF (defined in Section  
\ref{sub:barrier_functions_and_shielding}): see Section \ref{sec:shieldnn} for 
formal details.


Unfortunately, neither of the components \textbf{C1} nor \textbf{C2} can 
operate effectively on the same raw sensor measurements, $y[n]$, used by the  
controller; this is especially true given our intention to implement them via 
ZBFs and controller shields. In particular, both require some (limited) state 
information about the ADS in order to perform their tasks. Thus, 
\myalg~requires a perception/estimator component to provide state information 
to \textbf{C1} and \textbf{C2}. Note that we deliberately exclude the design of 
such an estimator from the \myalg~framework in order to provide additional 
flexibility: in particular, since the controller $\nn_c$ may effectively 
contain an estimator, we wish to allow for estimation to be offloaded, too -- 
provided it is executed locally just-in-time before informing \textbf{C1} and 
\textbf{C2} (see Section \ref{sub:analysis_of_energyshield_offloading}). 
Nevertheless such an estimator is necessary for \myalg, so we include it as 
component:
\begin{enumerate}[label={\bfseries C\arabic*:}]
	\setcounter{enumi}{2} %
	\item \textbf{State Estimator.} \myalg~requires (minimal) state 
		estimates be provided as input to \textbf{C1} and \textbf{C2}. By 
		convention, this estimator will be a NN denoted by $\nn_p : y  \mapsto  
		x$. \emph{We assume that $\nn_p$ can be computed by on-vehicle  
		hardware in one sample period.}
\end{enumerate}
\noindent The interface of \textbf{C3} with both \textbf{C1} and \textbf{C2} 
makes the latter two components (state-)context aware. \emph{That is, 
\myalg~makes {\bfseries context-aware offloading decisions} based on the 
current vehicle state.} Moreover, it is important to note that since the prior 
control action will be held during offload $\Delta_\text{max}$, the output of 
\textbf{C2}, control in addition to state dependent: that is, \textbf{C2} 
actually produces an output $\Delta_\text{max}(\hat{x}, u)$ for  (arbitrary) 
state $\hat{x}$ and the control $u$ applied just before offload.

\myalg~has one further important component, but one that is motivated purely by 
energy savings with no effect on safety. Crucially the known expiration of 
safety provided by \textbf{C2}, i.e. $\Delta_\text{max}(x)$, affords the 
opportunity to use additional information in making an offload decision. In 
particular, an estimate of the anticipated edge response time, $\hat{\Delta}$, 
can be used to \emph{forego} offloads that are unlikely to complete before the 
expiration of the safety deadline, $\Delta_\text{max}(x)$. For this reason, 
\myalg~contains an estimator of edge response time to preemptively skip 
offloads that are likely to fail:
\begin{enumerate}[label={\bfseries C\arabic*:}]
	\setcounter{enumi}{3} %
	\item \textbf{Edge-Response Estimator.} \myalg~specifies that an 
		estimate of the current edge response time, $\hat{\Delta}$, is provided 
		to inform offloading decisions.
\end{enumerate}
We note, however, that \emph{\myalg~doesn't specify a particular estimator to 
be used in this component}: any number of different estimators may be 
appropriate, and each estimator may lead to different energy profiles depending 
on its efficacy. Moreover, since $\hat{\Delta}$ is never used to override 
$\Delta_\text{max}(x)$, safety is preserved irrespective of the specific 
estimator used.

\begin{figure}[!tbp]
\centering
{\includegraphics[,width = 0.38\textwidth]{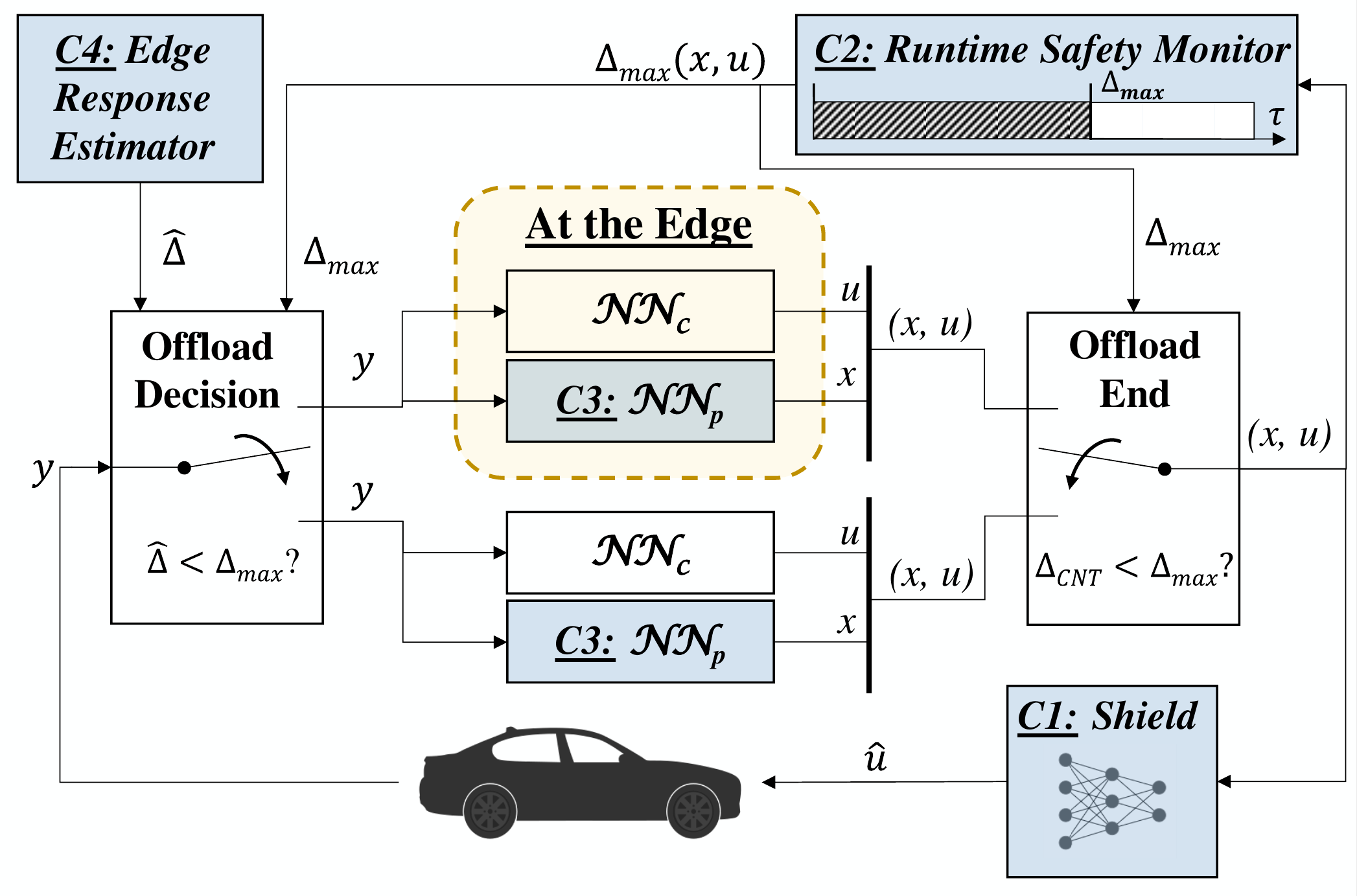}}
\vspace{-2ex}
\caption{EnergyShield Framework}
\label{fig:EnergyShield_block}
\end{figure}

The interconnection of the components \textbf{C1} through \textbf{C4} in 
\myalg~is illustrated in Figure \ref{fig:EnergyShield_block}. Note that component \textbf{C3}, the 
state estimator, is connected to components \textbf{C1} and \textbf{C2}, the 
controller shield and safety runtime monitor, respectively. Also note that the 
output of \textbf{C2} provides a signal $\Delta_\text{max}(x)$ to the 
offloading decision switch; also informing that decision is the estimate of 
immediate edge-response times provided by component \textbf{C4}.


\subsection{Semantics of an \myalg~Offloading Decision} 
\label{sub:analysis_of_energyshield_offloading}
In this subsection, we consider the timeline of a single, hypothetical 
\myalg~offloading decision to illustrate the details of the interplay between 
the components described in Section \ref{sub:energyshield_structure}. In 
particular, suppose that an offloading interval has just been completed, and at 
time index $n_0$ a new offloading decision is to be taken.

We call the time between the initialization of an offloading decision and the 
time that offloading decision has been resolved an \textbf{\itshape  offloading 
period} (the resolution is either by a response from the edge or a fail-over to 
on-vehicle compute). \\[-5pt]

\noindent \emph{Timeline:}\\[-5pt] %
\hrule %
\vspace{2pt} %
\hrule \vspace{2pt}

\noindent \timeIdx{n_0-1} \underline{\emph{Last time index of previous offload 
period.}} %

\begin{itemize}
	\item Assumption: $\hat{x}[n_0] \triangleq x[n_0 - 1] = \nn_p( y[n_0 - 
		1])$ is always computed locally in the last sample of the previous 
		offloading period.
\end{itemize}

\noindent \timeIdx{n_0} \underline{\emph{Initial time index of new offload  
period.}}

\begin{itemize}
	\item The first sample of the new offload period inherits a locally 
		computed $\hat{x}[n_0]$ from the previous offloading period.

	\item $\hat{x}[n_0]$ is provided to \textbf{C1} to correct $u[n_0]$ as 
		calculated by the previous offloading. Let this correct control action 
		be $\hat{u}[n_0]$

	\item Fix control action applied at $n_0$: i.e. $u_\text{0-h} =  
		\hat{u}[n_0]$, with $\hat{u}[n_0]$ as calculated above.

	\item $\hat{x}[n_0]$ and $u_\text{0-h}$ are provided to \textbf{C2}, the 
		runtime safety monitor, to generate  
		$\Delta_\text{max}(\hat{x}[n_0],u_\text{0-h})$

	\item \textbf{C4} generates an estimate for the edge response time,  
		$\hat{\Delta}$ based on all packets exchanged so far.

	\item \textbf{\itshape Offloading Decision:} 

		\begin{itemize}
			\item \emph{Proceed with offload, if:} $\hat{\Delta} \leq 
				\Delta_\text{max}(\hat{x}[n_0],u_\text{0-h})$ \textbf{AND} \\ 
				$\Delta_\text{max}(\hat{x}[n_0],u_\text{0-h}) \geq 1$; i.e., 
				proceed to transmit $y[n_0]$ to the edge. Initialize offload 
				duration counter:  $\Delta_\text{cnt} = 1$

			\item \emph{Otherwise, terminate offload period, and use local 
				fail-safe.} Skip to \textbf{\itshape Unsuccessful  Offload} 
				with 
				$\Delta_\text{max}(\hat{x}[n_0],u_\text{0-h})\negthickspace =  
				\negthickspace 1$.
		\end{itemize}
\end{itemize}

\noindent \timeIdx{\vdots} ~\\

\noindent \timeIdx{n_0 + \Delta_\text{cnt}} \underline{\emph{Offload in  
progress; no edge response and $\Delta_\text{cnt} < $}} %

\hspace{38pt} \underline{\emph{$\Delta_\text{max}(\hat{x}[n_0])$}}

\begin{itemize}
	\item Maintain zero-order hold of $u[n_0 + \Delta_\text{cnt}] =  
		u_\text{0-h}$.

	\item Increment $\Delta_\text{cnt}$: $\Delta_\text{cnt} \leftarrow  
		\Delta_\text{cnt} + 1$.
\end{itemize}

\noindent \timeIdx{\vdots} ~\\

\noindent \textbf{Now, the current offload period ends in one of two ways:}

\noindent \textbf{\itshape I) Successful Offload:} \emph{(resume timeline from 
$n_0 + \Delta_\text{cnt}$ )} %

\noindent \timeIdx{\vdots} ~\\

\noindent \timeIdx{n_0 + \Delta} \underline{\emph{Edge response received;  
$\Delta = \Delta_\text{cnt} \leq \Delta_\text{max}(\hat{x}[n_0],u_\text{0-h})$}} %

\begin{itemize}
	\item Apply received control action $u[n_0 + \Delta] = \nn_c(y[n_0])$.

	\item Initiate local evaluation of $\nn_c$ for next time interval.

	\item Initiate local evaluation of $\nn_p$ for next time interval.

	\item $n_0 - \Delta$ becomes time $n_1 - 1$ for the starting index of 
		the next offload period. (See $n_0 - 1$ time index above.)
\end{itemize}

\noindent \textbf{\itshape II) Unsuccessful Offload:} \emph{(resume timeline 
from $n_0 + \Delta_\text{cnt}$)} %

\noindent \timeIdx{\vdots} ~\\

\noindent \timeIdx{n_0 + {\scriptstyle \Delta_\text{max}}} \underline{\emph{No 
edge response received, and safety expired;}} %

\hspace{38pt} \underline{\emph{ $\Delta_\text{cnt} = 
\Delta_\text{max}(\hat{x}[n_0],u_\text{0-h})$}}

\begin{itemize}
	\item Maintain original zero-order hold control \\
	$u[n_0 + 
		\Delta_\text{max}(\hat{x}[n_0],u_\text{0-h})] = u_\text{0-h}$ (there 
		have been no  updates.)

	\item Initiate local evaluation of $\nn_c$ for next time interval.

	\item Initiate local evaluation of $\nn_p$ for next time interval.

	\item $n_0 - \Delta_\text{max}(\hat{x}[n_0])$ becomes time $n_1 - 1$ for 
		the starting index of the next offload period. (See $n_0 - 1$ time 
		index above.) 
\end{itemize}

\hrule \vspace{2pt} \hrule

~\\[-5pt]

In particular, note two crucial facts. First, if \textbf{C2} returns 
$\Delta_\text{max}(\hat{x},u) = 0$, then it effectively forces pure on-vehicle 
evaluation of $\nn_c$ and $\nn_p$. Second, we ensured that an up-to-date 
estimate of the state is always available for both \textbf{C1} and \textbf{C2} 
before they have to act.



%

%

\section{\myalg: Provably Safe Offloading} 
\label{sec:shieldnn}

\subsection{Main Formal Result} 
\label{sub:formal_result}

\subsubsection{Formal Assumptions} 
\label{ssub:formal_assumptions}
We begin this section with a list of formal assumptions about the ADS. These 
are largely based around the structure of \myalg, as described in Section 
\ref{sec:framework}.

\begin{assumption}[ADS Safety]
\label{as:ads_safety}
Consider a fixed point in the plane as a stationary obstacle to be avoided by the ADS, and a 
disk or radius $\bar{r}$ around the origin to be a set of unsafe states; see 
Figure \ref{fig:obstacle_diagram}.
\end{assumption}

\begin{assumption}[ADS Model]
\label{as:ads_model}
Let Assumption \ref{as:ads_safety} hold. Thus, suppose that the ADS vehicle is 
modeled by the KBM dynamics in \eqref{eq:kbm_dynamics}. Suppose further that 
interactions with this model happen in discrete time with zero-order hold. Let 
$T$ be the sampling period. 
\end{assumption}

\begin{assumption}[ADS Sensors]
\label{as:ads_sensors}
Let Assumptions \ref{as:ads_safety} and \ref{as:ads_model} hold. Suppose the  
KBM-modeled ADS has access to samples of a sensor signal, $s[n] \in 
\mathbb{R}^N$, at each discrete time step, and there is a perception NN, 
$\nn_\text{p} :  s[n] \mapsto \chi[n]$, which maps the sensor signal at each 
discrete time to the (exact) KBM state at the same time instant, $\chi[n]$.
\end{assumption}

\begin{assumption}[ADS Control]
\label{as:ads_control}
Let Assumptions \ref{as:ads_safety} - \ref{as:ads_sensors} hold. Suppose this 
KBM-modeled ADS vehicle has a NN controller, $\nn_c : s \mapsto 
\omega$, which at each sample has access to the sensor measurement $s$.
\end{assumption}

\begin{assumption}[ADS Local Computation]
\label{as:ads_compute}
Let Assumptions \ref{as:ads_safety} - \ref{as:ads_control} hold. Suppose that 
the output of $\nn_p$ and $\nn_c$ can be computed by ADS on-vehicle hardware in 
less than $T$ seconds -- i.e., not instantaneously. Thus, suppose that the 
control action is obtained with a one-sample computation delay when using 
on-vehicle hardware: i.e., the control action applied at sample $n+1$ is 
$\omega[n+1] = \nn_c(s[n])$.
\end{assumption}


\subsubsection{Component Design Problems} 
\label{ssub:component_design_problems}
There are two central problems that need to be solved: i.e., corresponding to 
the design of the two main components of \myalg, \textbf{C1} and  \textbf{C2} 
(see Section \ref{sec:framework}).

The solutions to these problems are deferred to Sections  
\ref{sub:es_kbm_controller_shield} and Section  
\ref{sub:es_kbm_run_time_monitor}, respectively. We state them here in order to 
facilitate the statement of our main result in the next subsection.

\begin{problem}[Controller Shield Design \textbf{(C1)}]
\label{prob:shield_design}
Let Assumptions \ref{as:ads_safety} - \ref{as:ads_compute} hold. Then the 
problem is to design: first, design functions $h$ and $\alpha$ such that they 
constitute a ZBF for the KBM (see); and then using this ZBF, design a 
controller shield, $\mathfrak{S}$ for the KBM model. The resulting controller 
shield must have the following additional property for \textbf{\itshape a 
discrete-time version of the KBM} with zero-order-hold inputs:
\begin{itemize}
	\item Let $\chi[n_0 - 1]$ and $\chi[n_0]$ be KBM states such that 
		$h(\chi[n_0 - 1]), h(\chi[n_0]) > 0$, and let $\chi[n_0]$ result from a 
		feasible input $\hat{\omega}[n_0-1]$ applied in state $\chi[n_0-1]$. 
		Then the control action
		\begin{equation}
			\hat{\omega}[n_0] = \mathfrak{S}(\chi[n_0 - 1], \omega[n_0])
		\end{equation}
		must yield a state $\chi[n_0 + 1]$ such that $h(\chi[n_0+1]) > 0$; 
		i.e., the controller shield preserves safety under discretization of 
		the KBM and one-step estimation delay (associated with $\nn_p$), as in 
		the case of no computations being offloaded.
\end{itemize}





\end{problem}

\begin{problem}[Runtime Safety Monitor Design \textbf{(C2)}]
\label{prob:safety_monitor}
Let Assumptions \ref{as:ads_safety} - \ref{as:ads_compute} hold, and assume 
that $h$, $\alpha$ and $\mathfrak{S}$ solve Problem \ref{prob:shield_design}. 
Then the problem is to design a runtime safety monitor:
\begin{equation}
	\Delta_\text{max} : \mathbb{R}^3 \times \Omega_\text{admis.} \rightarrow \mathbb{N}
\end{equation}
with the following property:

\begin{itemize}
	\item Let $\chi[n_0-1]$ be such that $h(\chi[n_0-1]) > 0$. Then for 
		constant control, $\omega = \omega[n_0]$, applied to the discretized 
		KBM starting from $\chi[n_0-1]$ the following is true:
		\begin{equation}\label{eq:discrete_runtime_safety_cond}
			\forall n = 0, \dots, \Delta_\text{max}(\chi[n_0-1], \omega[n_0]) ~ . ~
			h( \chi[n_0 - 1 + n] ) > 0
		\end{equation}
		i.e. the constant control $\omega = \omega[n_0]$ preserves safety for 
		at least $\Delta_\text{max}(\chi[n_0-1], \omega[n_0])$ samples  
		from state $\chi[n_0-1]$.
\end{itemize}
(The delay in $\chi[n_0-1]$ accounts for the computation time of  $\nn_p$.)
\end{problem}


\subsubsection{Main Result} 
\label{ssub:main_result}
Given solutions to these problems, it is straightforward to conclude our main 
formal result.
\begin{theorem}
	Let Assumptions \ref{as:ads_safety} - \ref{as:ads_compute} hold, and  
	assume a ZBF for the KBM dynamics, using which Problem 
	\ref{prob:shield_design} and Problem \ref{prob:safety_monitor} can be 
	solved.

	Then the offloading policy described in Section 
	\ref{sub:analysis_of_energyshield_offloading} preserves safety of the 
	KBM-modeled ADS (Assumptions \ref{as:ads_safety} and \ref{as:ads_model}).
\end{theorem}

\begin{proof}
	The proof follows largely by construction. Each offload period is limited 
	in duration by the runtime safety monitor; thus, a safety monitor that 
	solves Problem \ref{prob:safety_monitor} will ensure safety under the 
	specified constant control action during the offload period. Then by the 
	additional property of the controller shield in Problem 
	\ref{prob:shield_design}, safety can be maintained after the offloading 
	period ends: i.e., either by a deciding to perform a new offload if there 
	remains significant safety margin, or by executing the shield in closed 
	loop because there is no offload safety margin.
\end{proof}

\begin{corollary}
	Let Assumptions \ref{as:ads_safety} - \ref{as:ads_compute} hold, and 
	consider the ZBF for the KBM dynamics specified in Section  
	\ref{sub:kbm_controller_shield}. Then the controller shield in Section 
	\ref{sub:es_kbm_controller_shield} uses this ZBF and solves Problem 
	\ref{prob:shield_design}; likewise, the runtime monitor in Section 
	\ref{sub:es_kbm_run_time_monitor} uses this ZBF and solves Problem 
	\ref{prob:safety_monitor}. Hence, our implementation of \myalg~is safe.
\end{corollary}





\subsection{KBM Controller Shield} 
\label{sub:es_kbm_controller_shield}
Fortunately, we have access to a preexisting ZBF and controller shield designed 
for the KBM: see Section \ref{sub:kbm_controller_shield} 
\cite{ferlez2020shieldnn}. That is, the ZBF is available after using the design 
methodology in \cite{ferlez2020shieldnn} to choose the parameter $\sigma$  (see 
Section \ref{sub:kbm_controller_shield}); for simplicity, we will omit further 
discussion this design process. Thus for this section, we refer to a fully 
implemented controller shield as $\mathfrak{S}_\text{KBM}$, with the 
understanding that it has been designed for the relevant KBM model and safety 
parameter $\bar{r}$ (see Figure \ref{fig:obstacle_diagram}); viz.  Assumptions 
\ref{as:ads_safety} and \ref{as:ads_model}.

Thus, the main requirement of this section is to alter 
$\mathfrak{S}_\text{KBM}$ as necessary so that it satisfies the additional 
property required in Problem \ref{prob:shield_design}. This is the subject of 
the following Lemma, which also defines the alterations in question.

\begin{lemma}
\label{lem:es_controller_shield}
Let Assumptions \ref{as:ads_safety} - \ref{as:ads_compute} hold as usual, and 
let $\mathfrak{S}_\text{KBM}$ be a controller shield designed under these 
assumptions as per Section \ref{sub:kbm_controller_shield}.

Then there exists a $\rho > 0$ such that the following controller shield:
\begin{equation}\label{eq:mod_shield}
	\mathfrak{S}_\text{KBM}^\rho : ((r, \xi, v), \omega) \mapsto
	\begin{cases}
		\mathfrak{S}_\text{KBM}((r \negthickspace -\negthickspace \rho, \xi, v), \omega) & \negthickspace  r\negthickspace -\negthickspace \rho \geq r_{\scriptscriptstyle\text{min}}(\xi) \\
		\beta_\text{max} & \negthickspace  r \negthickspace -\negthickspace  \rho \geq r_{\scriptscriptstyle\text{min}}(\xi) \wedge \xi \geq 0 \\
		-\beta_\text{max} & \negthickspace r \negthickspace -\negthickspace \rho \geq r_{\scriptscriptstyle\text{min}}(\xi) \wedge \xi < 0
	\end{cases}
\end{equation}
solves Problem \ref{prob:shield_design}.
\end{lemma}
\noindent The proof of this Lemma is deferred to Appendix 
\ref{sec:proof_lemma_1}.

A further remark is in order about Lemma \ref{lem:es_controller_shield}. Note 
that the altered controller shield $\mathfrak{S}_\text{KBM}^\rho$ maintains the 
energy efficient implementation of the controller shield 
$\mathfrak{S}_\text{KBM}$ as designed in \cite{ferlez2020shieldnn}; indeed only 
two additional ReLU units are required, which amounts to a trivial 
per-evaluation increase in energy consumption.


\subsection{KBM Runtime Safety Monitor: Generating Safety Expiration Times} 
\label{sub:es_kbm_run_time_monitor}
Recall that the runtime safety monitor of \myalg~must provide an 
\emph{expiration} on the safety of the vehicle during an offload period, 
throughout which only a single, fixed control input is applied. This expiration 
must come with a \emph{provable guarantee} that the vehicle safety is not 
compromised in the interim. In the formulation of \myalg~and Problem 
\ref{prob:safety_monitor}, this means only that $h_{\bar{r},\sigma}$  must 
remain non-negative until the expiration of the deadline provided by the 
runtime safety monitor: see the condition 
\eqref{eq:discrete_runtime_safety_cond} of Problem \ref{prob:safety_monitor}.

This formulation is convenient because it means that the problem can again be 
analyzed in continuous time, unlike our consideration of Problem 
\ref{prob:shield_design} above: the conversion back to discrete time involves a 
floor operation; and compensating for the one-sample state delay induced by 
computing $\nn_p$ involves subtracting one sample from the result. That is, to 
solve Problem \ref{prob:safety_monitor} and design an \myalg-safety monitor, it 
is sufficient provide a (real) time, $\nu$, s.t.:
\begin{equation}\label{eq:continuous_safety_cond}
	\forall t \in [0,\nu] \;.\; h \big( \zeta_{\scriptscriptstyle \mathbf{1}_{\omega[n_0]}}^{\scriptscriptstyle 0,\chi\negthickspace[n_0\negthinspace -\negthinspace 1]}(t) \big) > 0.
\end{equation}
In other words, the flow of $f_\text{KBM}$ started from $\chi[n_0-1]$ and using 
constant control $\omega[n_0]$ maintains $h > 0$ for the duration $[0, \nu]$. 
We emphasize again that such a $\nu$ can be converted into the correct sample 
units expected for the output $\Delta_\text{max}(\chi[n_0-1],\omega[n_0])$ by 
using a floor operation and subtracting one. With this context, we have the 
following Lemma, which solves Problem \ref{prob:safety_monitor}.

\begin{lemma}
\label{lem:es_safety_monitor}
Let Assumptions \ref{as:ads_safety} - \ref{as:ads_compute} hold as usual. Let 
\begin{equation}\label{eq:safety_monitor_formal}
	\Delta_\text{max}(\chi[n_0 - 1, \omega[n_0]) \triangleq  \max( \left\lfloor \nu(\chi[n_0-1],\omega[n_0]) \right\rfloor - 1, 0)
\end{equation}
where $\nu = \nu(\chi[n_0-1],\omega[n_0])$ solves the equation:
\begin{equation}\label{eq:gronwall_conclusion_statement}
	\sqrt{2} \cdot L_{h_{\bar{r},\sigma}} \cdot \lVert f_{\scriptscriptstyle\text{KBM}}(\chi[n_0-1],\omega[n_0]) \rVert_2 \cdot \nu \cdot e^{L_{f_\text{KBM}} \cdot \nu} = h( \chi[n_0 - 1] )
\end{equation}
for $L_{h_{\bar{r},\sigma}}$ and $L_{f_{\scriptscriptstyle\text{KBM}}}$ upper 
bounds on the Lipschitz constants of $h_{\bar{r},\sigma}$ and 
$f_{\scriptscriptstyle\text{KBM}}$, respectively.

Then $\Delta_\text{max}(\chi[n_0 - 1, \omega[n_0])$ so defined solves Problem 
\ref{prob:safety_monitor}.
\end{lemma}
\noindent The proof of Lemma \ref{lem:es_safety_monitor} is deferred to Appendix \ref{sec:proof_lemma_2}.

Lemma \ref{lem:es_safety_monitor} specifies a complete solution to Problem 
\ref{prob:safety_monitor}, as claimed. However in its immediate form, it 
requires numerically solving \eqref{eq:gronwall_conclusion_statement} with each 
evaluation of $\Delta_\text{max}( \chi[n_0-1], \omega[n_0] )$; i.e., each time 
a safety expiration time is requested from the runtime safety monitor (every 
sample in the case where the offloading period is terminated before offload).  
The nature of \eqref{eq:gronwall_conclusion_statement} is such that solving it 
numerically is not especially burdensome -- especially compared to the NN 
evaluations it replaces; however, as a practical matter it is also possible to 
implement it as a LUT to achieve greater energy efficiency.



%

%

%

\section{Experiments and Findings}

The purpose of this section is to assess the following key aspects of EnergyShield: (\emph{i}) the extent of energy savings achievable compared to the conventional continuous local execution mode given a wide variety of network conditions and server delays, (\emph{ii}) its ability to enforce the safety through obstacle collision avoidance, (\emph{iii}) how representative the upper bounds of the edge response time ($\Delta_{max}$) are of the inherent risks existing in the corresponding driving context, and (\emph{iv}) its generality when integrated with controllers of distinctive behavioral characteristics, in particular when it comes to different learnt driving policies.

\subsection{Experimental Setup}

\begin{figure}[!tbp]
\centering
{\includegraphics[,width = 0.49\textwidth]{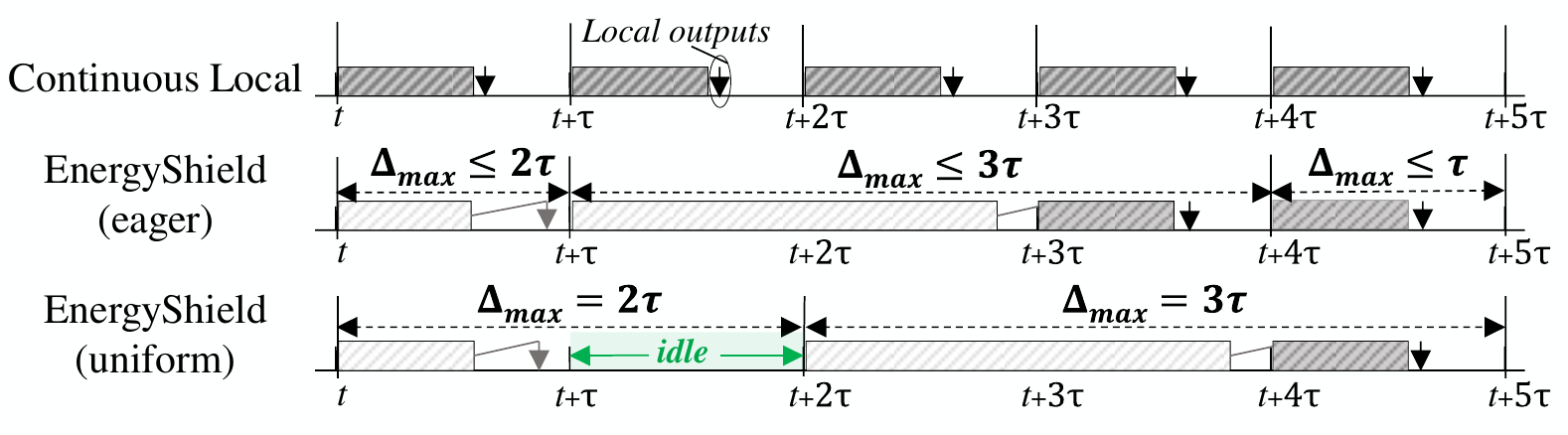}}
\vspace{-5ex}
\caption{The operational policies in our experiments given base time window $\tau$. Darker instances imply local execution. }
\label{fig:offloading_policies}
\vspace{-1ex}
\end{figure}

\subsubsection{Operational Policies: }In addition to the baseline continuous local execution, we designate two EnergyShield offloading modes: 
\begin{itemize}
    \item \textit{\textbf{Eager:}} in which a new offloading interval is immediately started if either edge responses have been received at the ADS or $\Delta_{max}$ expired.
    \item \textit{\textbf{Uniform:}} in which the start of a new offloading interval is always delayed until $\Delta_{max}$ expires, regardless of whether edge responses have been received or not.
\end{itemize}
We define both these modes to reflect the attainable behavioral trade-offs of EnergyShield with regards to realizing an ideal control behavior or maximizing energy efficiency.  This distinction is illustrated through the first offloading interval in Figure \ref{fig:offloading_policies} in which the uniform EnergyShield idles upon its retrieval of the edge responses until $\Delta_{max}$ expires unlike the eager EnergyShield mode.



\subsubsection{Experimental Scenario:} We perform our experiments using the CARLA open-source simulator for autonomous driving research \cite{dosovitskiy2017carla}. We follow the setup proposed in \cite{ferlez2020shieldnn}, and implement a similar experimental scenario\footnote{https://github.com/shieldNN/shieldNN2020}. Basically, the scenario involves a four-wheeled vehicle travelling from a starting position A to destination B along a 100m motorway track with 4 pedestrian obstacles in its path. The first obstacle spawns after 40m of the track, while the remaining spawning positions are uniformly spaced between the first obstacle's position and that of the final destination -- with a potential $\pm10m$ variation along the longitudinal axis. 

\subsubsection{Experimental Settings: }Throughout this section, \textit{all} of our experiments are conducted under different combinations of the following two binary configuration parameters: 
\begin{itemize}
    \item \textbf{S}: this binary variable indicates whether the Controller Shield component is \textit{active} (see Section \ref{sub:energyshield_structure}).
    \item \textbf{N}: this binary variable indicates whether this is a more challenging, ``noisy'' version of the experimental test case.
\end{itemize}
In particular, the noisy version entails perturbing the obstacles' spawning positions by additional values sampled from a normal distribution $\mathcal{N}(0,1.5m)$ along both the longitudinal and latitudinal axis. For example, the configuration settings (S = 1, N = 0) indicate that the experiment was performed with Controller Shield active and with no perturbations in the obstacles' spawning positions.



\subsubsection{Simulation Setup: }For the controller model, its first stage entails two concurrent modules: an object detector as the large NN model of the ADS and a $\beta$ Variational Autoencoder ($\beta$-VAE) providing additional latent feature representations of the driving scene. Both components operate on 160$\times$80 RGB images from the vehicle's attached front-facing camera. In the subsequent stage, a Reinforcement Learning (RL) agent aggregates the detector's bounding box predictions, latent features, and the inertial measurements ($\delta_f^c$, $v$, and $a$) to predict vehicle control actions (steering angle and throttle). The inertial measurements can be fetched directly from CARLA, whose positional and orientation measurements are also used directly to calculate $r$ and $\xi$ relative the vehicle's current nearest obstacle for obstacle state estimation. 
We trained the RL controller agents using a reward function, $\mathcal{R}$, that aims to maximize track completion rates through collision avoidance while minimizing the vehicle's center deviance from the primary track. For the definition of $\mathcal{R}$ as well as details of the RL agents see Appendix \ref{sec:experimental_details}.


\subsubsection{Performance Evaluations: } We use a pretrained ResNet-152 for our object detector and benchmark its performance in terms of latency and energy consumption when deployed on the industry-grade Nvidia Drive PX2 Autochauffer ADS. We found that a single inference pass on the ResNet-152 took $\approx$ 16 ms, and accordingly, we fixed the time-step in CARLA at 20 ms since the detector-in-the-loop was the simulation's computational bottleneck. To evaluate the wireless transmission power, we use the data transfer power models in \cite{huang2012close} and assume a Wi-Fi communication link.

\subsubsection{Wireless Channel Model: }We model the communication overheads between the ego vehicle and edge server as:
\begin{equation}
    L_{comm} = L_{Tx} + L_{que}\;\; s.t., \; L_{Tx} = \tfrac{data\_size}{\phi},
\end{equation}
where $L_{que}$ represents potential queuing delays at the server whereas $L_{Tx}$ is the transmission latency defined by the size of the transmission data, $data\_size$, over the experienced channel throughput, $\phi$. Here, we assume $\phi$ as the ``effective'' channel throughput experienced at the ego vehicle, which takes into consideration potential packet drops, retransmissions, etc. We leverage a Rayleigh distribution model as in \cite{odema2022testudo} to sample throughput values $\phi \sim Rayleigh(0,\sigma_{\phi})$  with zero mean and $\sigma_{\phi}$ scale (default $\sigma_{\phi}$=20 Mbps). Details on the modelling of queuing delays, $q$, and the server response time estimation (i.e., $\hat{\Delta})$ are provided in Appendix \ref{sec:experimental_details}.

\subsection{EnergyShield Evaluations}

The purpose of this experiment is to assess the controller's performance when supplemented with EnergyShield in terms of energy efficiency and safety. For every configuration of S and N, we run the test scenario for 35 episodes and aggregate their combined results. 


\textbf{Energy Efficiency:} We first assess the energy efficiency gains offered by EnergyShield compared to the baseline continuous local execution. As illustrated Figure \ref{fig:energy_safety}, the left barplot demonstrates that both modes of EnergyShield substantially reduce the energy consumption footprint of the NN compared to local execution across all S and N configurations. For instance, under the default configuration (S = 0, N = 0), EnergyShield energy reductions reach 20\% and 40.4\% for the eager and uniform modes, respectively. These numbers further improve for the subsequent configurations in which N = 1 or S = 1. Upon inspection, we find that this is the result of the ego vehicle encountering more instances in which obstacles are not in the direct line-of-sight of its heading. The reasons being that at N = 1, some obstacles can be displaced out of the primary lane that the ego vehicle follows to complete the track, whereas at S = 1, such instances result from the Controller Shield applying corrective behaviors on the NN's predicted steering outputs, resulting in more tangential orientations of the vehicle with respect to the obstacles (i.e., $\xi \to \pm \pi/2$). Accordingly, large values of $\Delta_{max}$ -- about 4-5 time samples (equivalent to 80-100 ms) -- are increasingly sampled, and that automatically translates into more offloading decisions. For instance at (S = 1, N = 0), the energy efficiency gains reach $24.3\%$ and $54.6\%$ for the respective eager and uniform modes. 

\textbf{Safety Evaluation: } To assess the EnergyShield's ability to enforce safety, we designate track completion rates (TCR \%) as a comparison metric to signify the proportion of times the vehicle was able to complete the track without collisions. Taking the local execution mode as the test scenario, the right barplot of Figure \ref{fig:energy_safety} shows that without an active Controller Shield (S = 0), collisions with the pedestrian obstacles cause the TCR\% to be 65.7\% at N = 0, and even less at 22.9\% for the noisy test case (N = 1). However, when the Controller Shield is active (S = 1), collisions are completely avoided and the TCR (\%) values jump to 100\% for both cases. This is also visible through the respective improvements in $\mathcal{R}$ which reached 13.3\% and 61.1\%. To further demonstrate such occurrences, we analyze in Figure \ref{fig:trajectories} the ego vehicle's chosen trajectories across 3 episodes of dissimilar (S, N) configurations. As shown, the (S = 0, N = 0) instance incurs a collision with the pedestrian object and does not complete the track. An active Controller Shield (S = 1), however, enforces a left or right corrective maneuvering action for obstacle avoidance and maintaining safety. Pictorial illustrations of these actions are also provided in the Figure. 

\begin{figure}[!tbp]
\centering
{\includegraphics[,width = 0.49\textwidth]{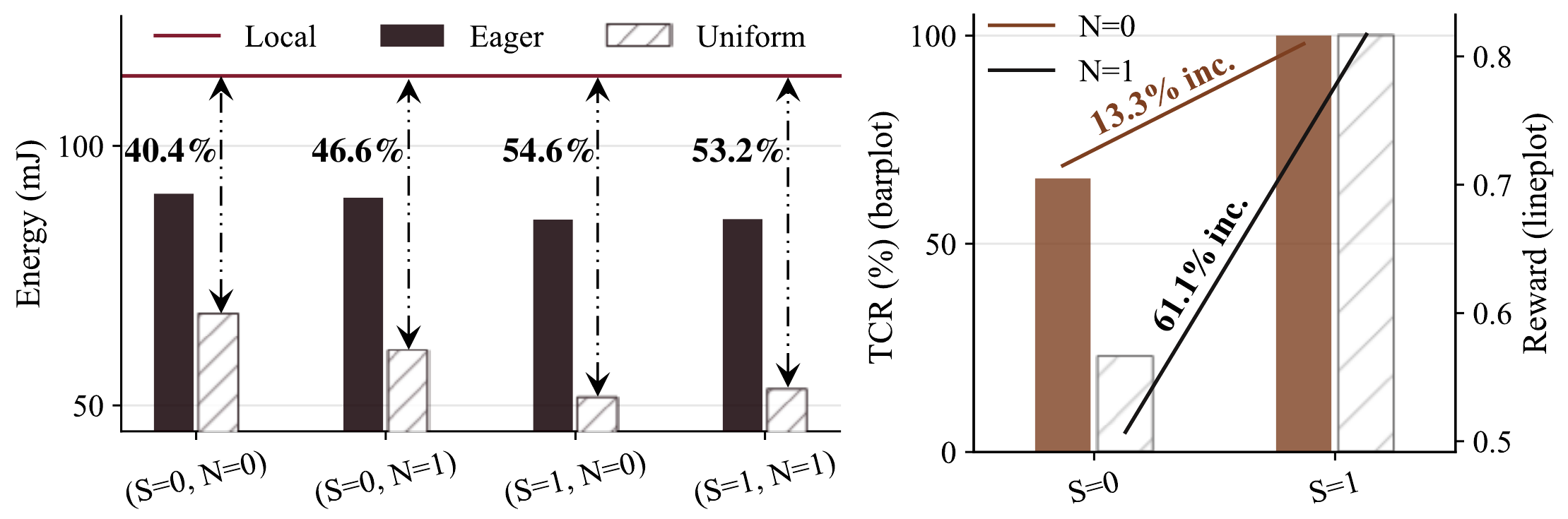}}
\vspace{-5ex}
\caption{EnergyShield's energy efficiency gains with respect to continuous local execution (\textit{left}) and safety analysis in terms of the $\mathcal{R}$ evaluation and \% TCR (\textit{right}).}
\label{fig:energy_safety}
\vspace{-2ex}
\end{figure}

\textbf{Energy vs. Distance: }To assess how representative the $\Delta_{max}$ upper bounds provided by the Runtime Safety Monitor are of the corresponding driving scene context, we examine EnergyShield's energy consumption at different distances from the nearest obstacle ($r$). The hypothesis is that larger $r$ values imply relatively ``safer'' driving situations, which would result in larger values of $\Delta_{max}$ to be sampled, and accordingly more offloading instances enhancing the NN's energy efficiency. As shown in Figure \ref{fig:E_v_d_combined}, we plot the average experienced normalized energy ratings of the two modes of EnergyShield with respect to local execution against $r$ across every configuration's set of 35 episodes. Each tick on the horizontal axis accounts for an entire range of 1m distances rather than a single value -- e.g., a value of 2 on the horizontal axis encompasses all distances in the range [2 - 3). At close distances ($r$ < 4m), we find that EnergyShield modes incur almost the same energy consumption overhead as that from the default local execution. This is mainly accredited to the Runtime Safety Monitor recognizing the higher risks associated with the close proximity from the objects, and accordingly outputting smaller values of $\Delta_{max}$ that can only be satisfied by local execution. As the distance from obstacles increases, so do the values of $\Delta_{max}$, causing a gradual increase in the number of offloading instances, followed by a progressive reduction in energy consumption. For instance, the eager and uniform modes achieve 32\% and 66\% respective reductions in energy consumption at $r = 13$ m for the (S = 1, N = 1) configuration. Even more so, all configurations of the respective eager and uniform modes at the ($r$ > 20m) bracket realize 33\% and 67\% respective energy gains.

\begin{figure}[!tbp]
\centering
{\includegraphics[,width = 0.46\textwidth]{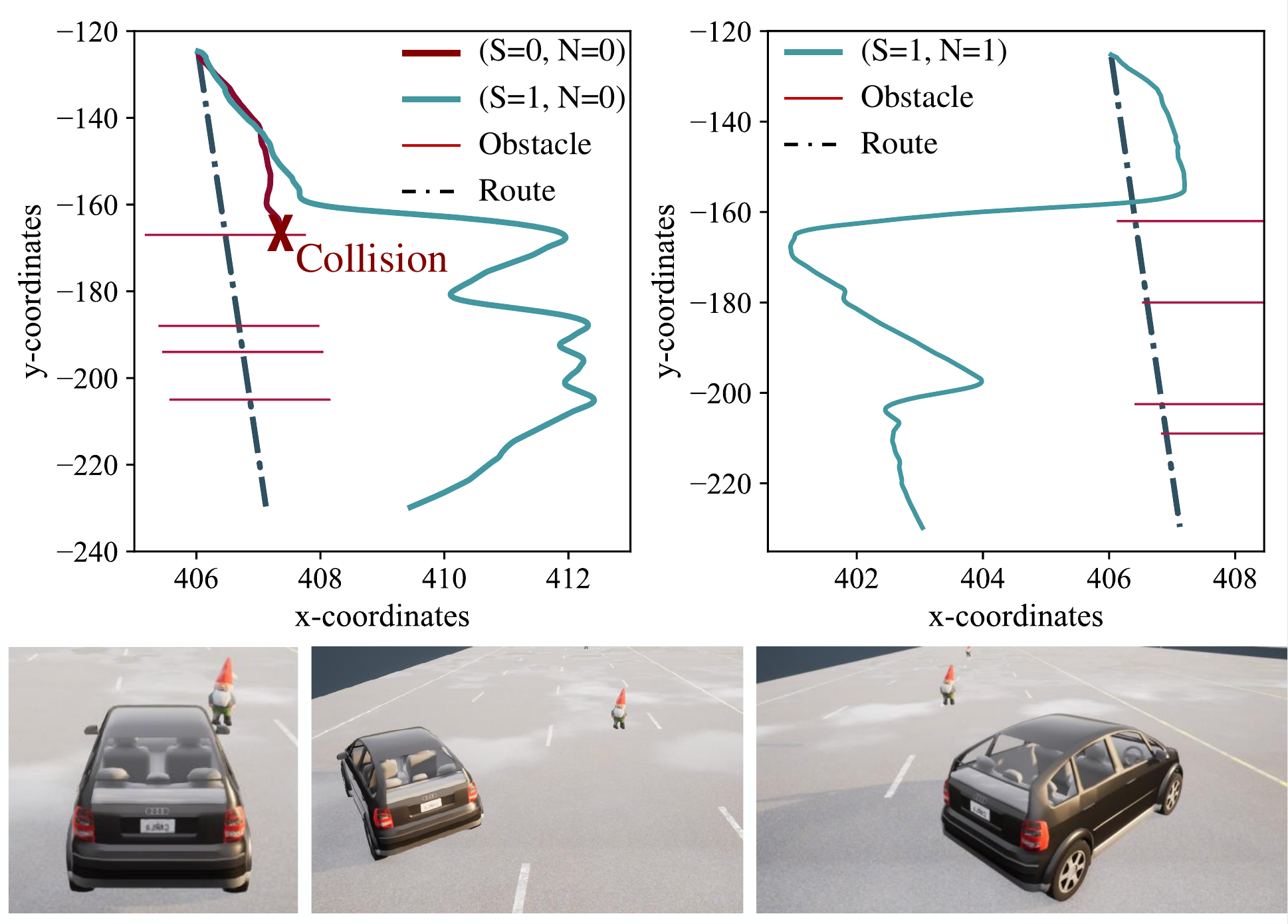}}
\vspace{-2ex}
\caption{Top: Example trajectories followed by the ego vehicle with the start point at the top. Bottom: illustration of how the ego vehicle under the aforementioned operational modes behaved in reaction to the first encountered obstacle.}
\label{fig:trajectories}
\end{figure}



\begin{figure}[!tbp]
\centering
{\includegraphics[,width = 0.45\textwidth]{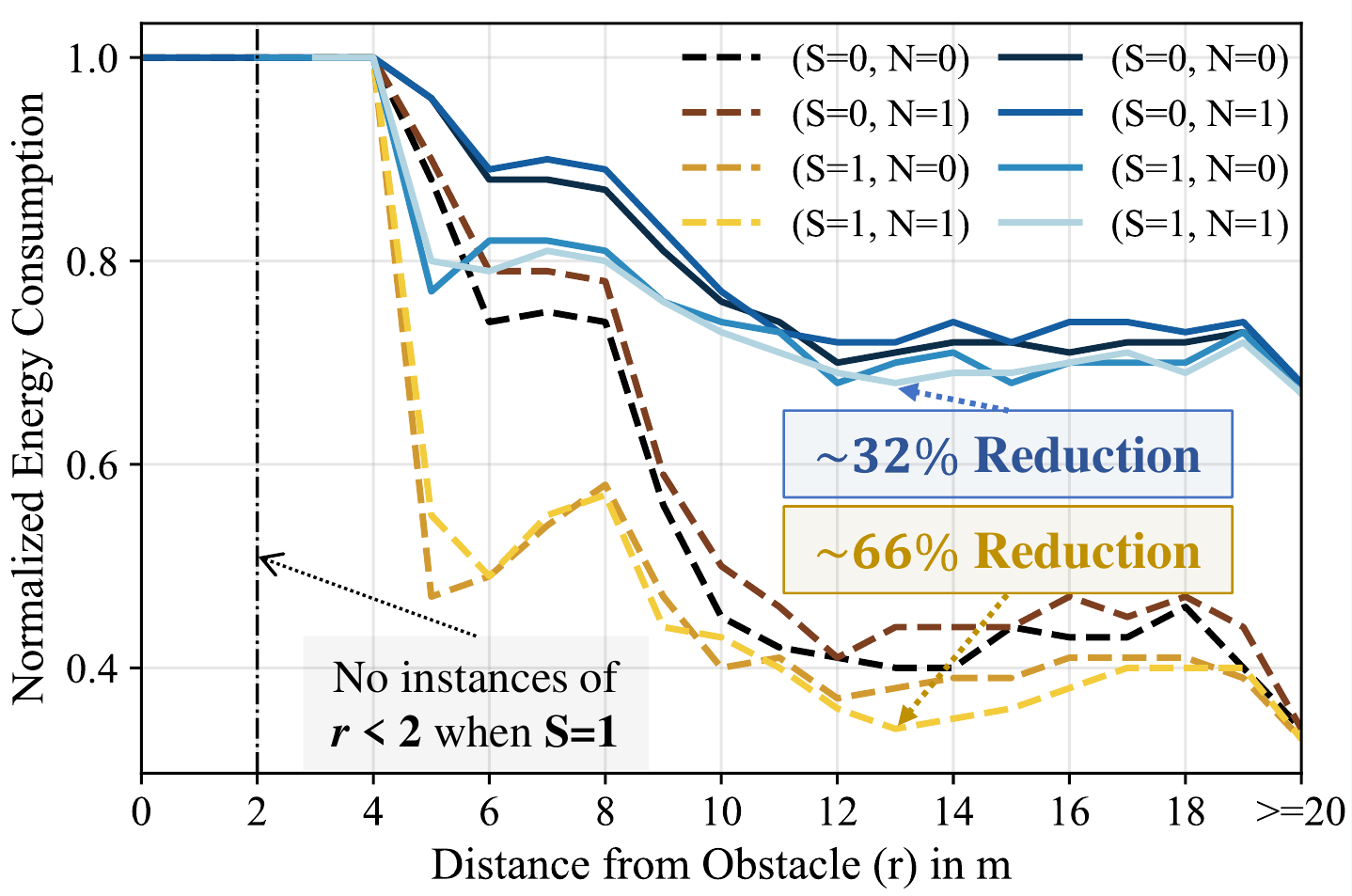}}
\vspace{-2ex}
\caption{Normalized Energy Gains for the eager (\textit{solid}) and uniform (\textit{dashed}) EnergyShield modes with respect to the distance from obstacle (\textit{r}) in m.}
\label{fig:E_v_d_combined}
\end{figure}

\subsection{Wireless Channel Variations}
In this experiment, we assess how the performance gains of EnergyShield are affected given variations of the wireless channel conditions. Specifically, given potential changes in the channel throughput, $\phi$, or the queuing delays, $q$, we investigate to what extent do the energy savings offered by EnergyShield vary. Additionally, we examine for every set of experimental runs what percentage of their total elapsed time windows were extra transition windows needed to complete a single offloading instance, which we denote by the \% \textit{Extra Transit Windows} metric. From here, we first analyze such effects when varying $\sigma_{\phi}\in \{20, 10, 5\}$ Mbps given a fixed $q=1$ ms, and then when varying $q \in \{10, 20, 50\}$ ms given a fixed $\sigma_{\phi}=10$ Mbps. For the uniform EnergyShield, we notice in Figure \ref{fig:bbox_windows} that the \textit{$\%$ Extra Transit Windows} drops for the contrasting conditions of high throughput ($\sigma_{phi}=20$ Mbps) and high queuing delays ($q=50$ ms), reaching medians of 7\% and 8\%, respectively. This can be justified in light of how the benign channel conditions ($\sigma_{\phi}=20$ Mbps) indicate that the majority of offloading instances are concluded in a single time window with no considerable need for extra transmission windows. Whereas at unfavorable wireless conditions ($q=50$ ms), $\hat{\Delta}$ values often exceed $\Delta_{max}$, leading EnergyShield to opt for local execution more often so as to avoid wireless uncertainty, lowering the total number of transmission windows alltogether. Such effects are also visible in the twin Figure \ref{fig:bbox_energy} as EnergyShield's energy consumption varies across these contrasting conditions, reaching respective medians of $45\%$ and $93\%$ of the local execution energy at $\sigma_{\phi}$=20 and $q$=50.  

\begin{figure}[!tbp]
\centering
{\includegraphics[,width = 0.46\textwidth]{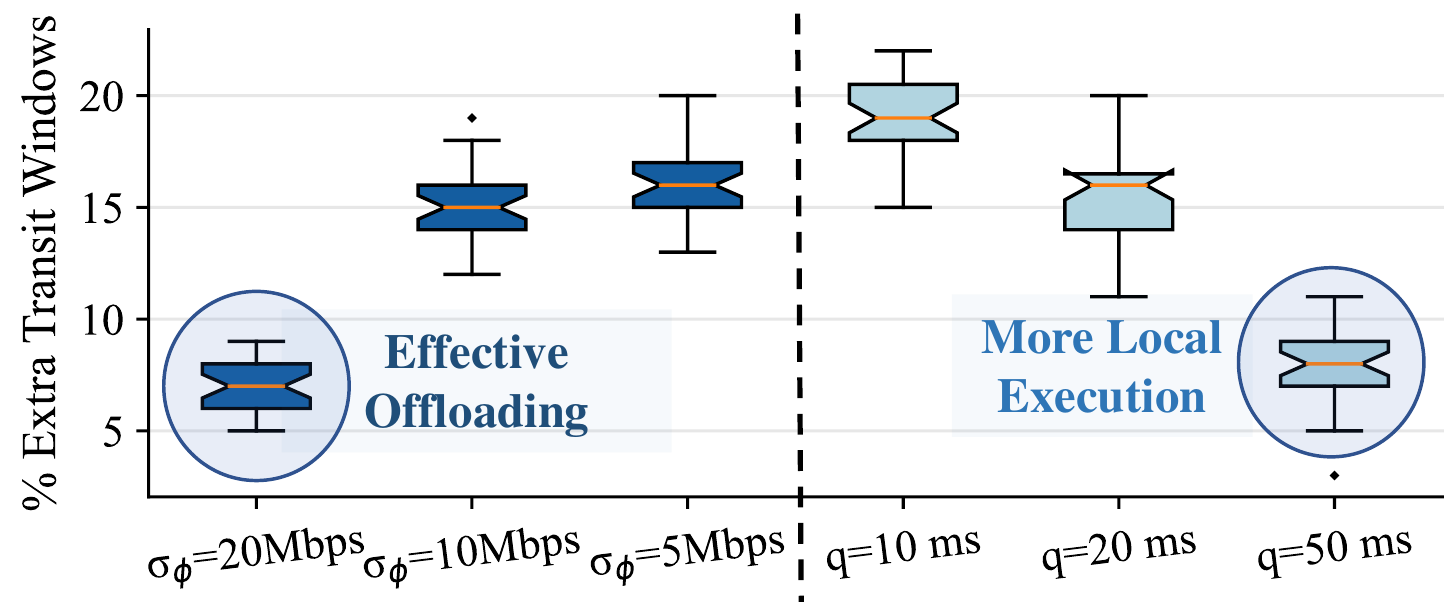}}
\vspace{-2.5ex}
\caption{Analyzing the \% extra transit windows over 35 episodes of uniform EnergyShield given various $\sigma_{\phi}$ and $q$.}
\label{fig:bbox_windows}
\end{figure}

\begin{figure}[!tbp]
\centering
{\includegraphics[,width = 0.46\textwidth]{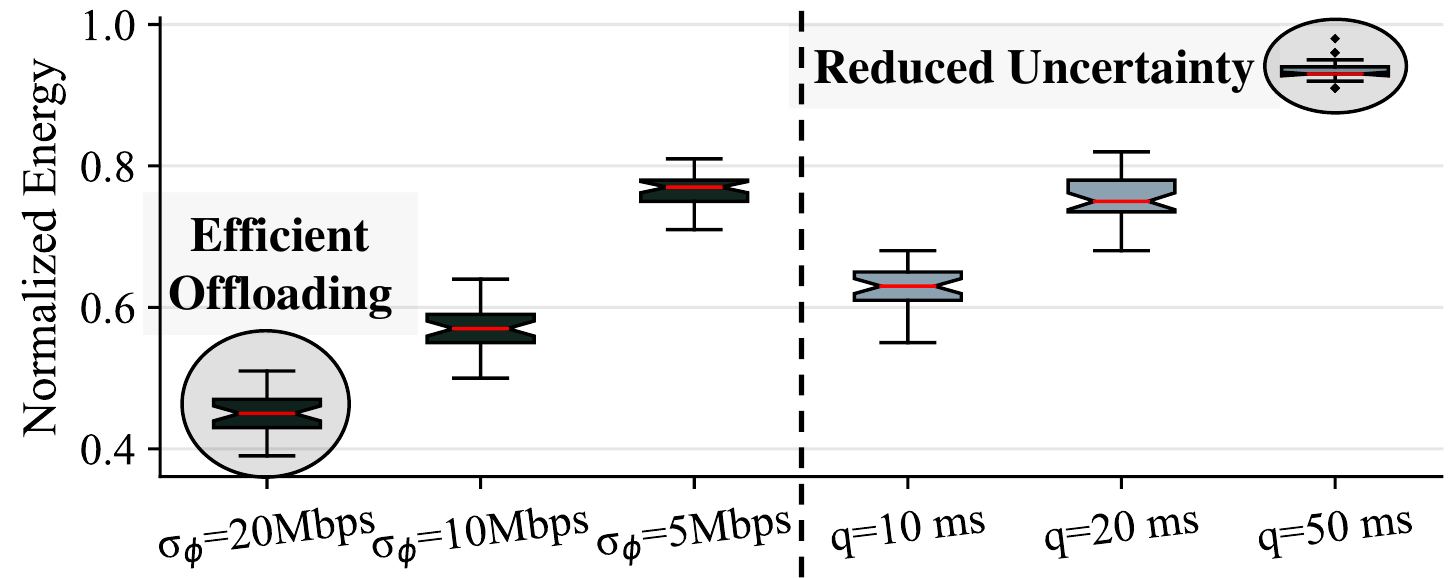}}
\vspace{-2.5ex}
\caption{Analyzing the Normalized Energy cons. over 35 episodes of uniform EnergyShield given various $\sigma_{\phi}$ and $q$.}
\label{fig:bbox_energy}
\end{figure}

\subsection{Generality}

To assess the generality of EnergyShield, we train 3 additional RL controllers to evaluate how consistent EnergyShield is with regards to maintaining safety guarantees, and how the energy efficiency gains would vary given a distinctive driving behavior for each agent. 
Hence, we repeat the primary test runs for the additional controllers in which we average the energy consumption and TCR across 35 episodes of each viable (S, N) configuration. We also report the average center deviance (CD) experienced by the ego vehicle from the primary track lane as a metric to characterize the different driving behaviors of each controller. In our experimental test case, larger CD imply larger $r$ values, that is, controllers with low values of CD tend to drive closer to the obstacles at higher risks of collisions (maximizing $\mathcal{R}$ through minimizing CD), whereas larger values of CD indicate the agents have learnt to take the farther lanes of the track in order to maximize track completions through prioritizing collision avoidance (maximizing $\mathcal{R}$ through maximizing TCR). In the following, we highlight the key findings and refer the reader to Table \ref{tab:agents} in Appendix \ref{sec:experimental_details} for the detailed results. 

Our first key finding is that given S = 0, TCR (\%) is dependent on CD ratings as only the controllers with CD > 5m consistently achieve the 100\% TCR. At S = 1, however, this dependency no longer holds at S = 1 when the Controller Shield enforces the obstacle avoidance safety guarantees, pushing all controllers to achieve 100\% TCR irrespective of CD values. The other interesting finding here is that across both modes of EnergyShield, the average energy consumption is less for controllers with larger CD ratings. For instance at (S = 1, N = 1), when the RL controllers are arranged in an increasing order of their CD values 2.8$\to$3.6$\to$5.4$\to$5.7 m, the average energy consumption per inference decreases in the respective order 53.1$\to$45.7$\to$42.1$\to$39.8 mJ. Indeed, this highlights EnergyShield's capability of conducting \textit{safe} and \textit{effective} context-aware offloading, especially given how the Runtime Safety Monitor provides large $\Delta_{max}$ to realize more energy gains in the safer situations (e.g., larger distances from obstacles), and how the Controller Shield always maintains safety guarantees independent of offloading decisions. 

\bibliographystyle{ACM-Reference-Format}
\bibliography{sources}

\appendix

%

\section{Proof of Lemma \ref{lem:es_controller_shield}} 
\label{sec:proof_lemma_1}

\begin{proof}
In this proof, let  $\chi[n^\prime] = (r[n^\prime], \xi[n^\prime], 
v[n^\prime])$,  $n^\prime \in \mathbb{Z}$.

The proof will be constructive. To this end, recall that we have assumed a 
capped (controlled) maximum vehicle velocity of $v_\text{max}$. Thus, let  
$\gamma = 2 \cdot v_\text{max} \cdot T$, where $T$ is the sample period (see  
Section \ref{sub:notation}). As a consequence, also note that for $r[n_0-1] > 
r_\text{min}(\xi[n_0-1])$:
\begin{equation}
	r[n_0+1] > r[n_0-1] - 2\cdot v_\text{max} \cdot T > \bar{r} - \gamma
\end{equation}
so that on any two-sample period
\begin{equation}
	| \dot{\xi} | \leq v_\text{max} \cdot \left(\tfrac{1}{\bar{r}-\gamma} + \tfrac{1}{\ell_r}\right)
\end{equation}
Then observe that the Lipschitz constant of the function $r_\text{min}$ is 
bounded by $L_{r_\text{min}} \leq \tfrac{\bar{r} \cdot  
\sigma}{2\cdot(1-2\cdot\sigma)^2}$. Finally, conclude that
\begin{multline}\label{eq:eta_def}
	| r_\text{min}(\xi[n_0-1]) - r_\text{min}(\xi[n_0+1]) |  \\
	\leq \tfrac{\bar{r} \cdot \sigma }{2\cdot(1-2\cdot\sigma)^2} \cdot v_\text{max} \cdot \left(\tfrac{1}{\bar{r}-\gamma} + \tfrac{1}{\ell_r}\right) \cdot 2 \cdot T \triangleq \eta
\end{multline}
Then choose $\rho \triangleq \eta + \gamma$.

Now, given the structure of the amended shield in \eqref{eq:mod_shield},  
establishing the conclusion of Problem \ref{prob:shield_design} can be broken  
into three cases: 
\begin{enumerate}[label={\itshape (\roman*)}]
	\item $r[n_0] \geq r_\text{min}(\xi[n_0]) + \rho$; (irrespective of the 
		position of $r[n_0-1]$)

	\item $r[n_0] < r_\text{min}(\xi[n_0]) + \rho$ and $r[n_0 -1] \geq 
		r_\text{min}(\xi[n_0 - 1]) + \rho$; 

	\item $r[n_0] < r_\text{min}(\xi[n_0]) + \rho$ and $r[n_0 -1] < 
		r_\text{min}(\xi[n_0 - 1]) + \rho$.

\end{enumerate}
In each of the three cases, we need to show that for the next state,  
$h(\chi[n_0 + 1]) > 0$, or equivalently $r[n_0 + 1] >  
r_\text{min}(\xi[n_0+1])$.

\emph{Case (i) and Case (ii).} The claim follows for these cases for 
essentially directly by the choice of $\rho$ above. In \emph{Case (ii)}, we 
have that
\begin{multline}\label{eq:case_1_2_main}
	r[n_0+1] - r_\text{min}(\xi[n_0+1]) \geq \big( r[n_0-1] - r_\text{min}(\xi[n_0-1]) \big)  \\
	- \big( r[n_0+1] - r_\text{min}(\xi[n_0+1]) - r[n_0-1] - r_\text{min}(\xi[n_0-1]) \big).
\end{multline}
From the above calculations, we see that the second term on the right-hand side 
of \eqref{eq:case_1_2_main} is bounded below by $-\rho$ (using the triangle  
inequality). Thus
\begin{equation}
	r[n_0+1] - r_\text{min}(\xi[n_0+1]) \geq \big( r[n_0-1] - r_\text{min}(\xi[n_0-1]) \big) - \rho
\end{equation}
and the desired conclusion follows since $r[n_0-1] - r_\text{min}(\xi[n_0-1])  
\geq \rho$ by assumption. A similar approach proves \emph{Case (i)}: simply 
repeat the calculations in the definition of $\rho$, only over one sample.

Thus, it remains to consider \emph{Case (iii)}. This case is somewhat easier, 
because the control signal is being overridden, so the state delay presents 
technical difficulties as above. Thus, it follows almost directly from the 
properties of the controller shield as designed in \cite{ferlez2020shieldnn}. 
In particular, the ShieldNN verifier establishes that the boundary between 
``safe'' and  ``unsafe'' controls is a concave (resp. convex) function of $\xi$ 
for $\xi \in [0, \pi]$ (resp. $\xi \in [-\pi,0]$). Hence, by  \cite[Theorem  
1]{ferlez2020shieldnn}, the constant control $\beta_\text{max}$ (resp.  
$-\beta_\text{max}$) always preserves safety for \emph{any duration of time} 
starting from a state $\xi \in [0,\pi]$ (resp. $\xi \in [-\pi,0]$).
\end{proof}


\section{Proof of Lemma \ref{lem:es_safety_monitor}} 
\label{sec:proof_lemma_2}

\begin{proof}
By the arguments above, the form of  $\Delta_\text{max}$ in  
\eqref{eq:safety_monitor_formal} will solve Problem \ref{prob:safety_monitor}, 
provided \eqref{eq:gronwall_conclusion_statement} implies 
\eqref{eq:continuous_safety_cond}. Thus, we confine the proof to showing this 
fact.

To begin, observe that:
\begin{multline}
	h_{\bar{r},\sigma}( \zeta_{\scriptscriptstyle \mathbf{1}_{\omega[n_0]}}^{\scriptscriptstyle 0,\chi\negthickspace[n_0\negthinspace -\negthinspace 1]}(t) )
	\geq
	\Big|
	\big|  
		h_{\bar{r},\sigma} ( \zeta_{\scriptscriptstyle \mathbf{1}_{\omega[n_0]}}^{\scriptscriptstyle 0,\chi\negthickspace[n_0\negthinspace -\negthinspace 1]}(t) )
		-
		h_{\bar{r},\sigma}( \chi[n_0-1] )
	\big| 
	\\
	-
	\big|
		h_{\bar{r},\sigma}( \chi[n_0-1] )
	\big|
	\Big|
\end{multline}
by the triangle inequality. Consequently:
\begin{multline}\label{eq:triangle_bound_pre_gronwall}
	\big|  
		h_{\bar{r},\sigma} ( \zeta_{\scriptscriptstyle \mathbf{1}_{\omega[n_0]}}^{\scriptscriptstyle 0,\chi\negthickspace[n_0\negthinspace -\negthinspace 1]}(t) )
		-
		h_{\bar{r},\sigma}( \chi[n_0-1] )
	\big| 
	\leq
	\big|
		h_{\bar{r},\sigma}( \chi[n_0-1] )
	\big|
	\\
	\implies
	h_{\bar{r},\sigma}( \zeta_{\scriptscriptstyle \mathbf{1}_{\omega[n_0]}}^{\scriptscriptstyle 0,\chi\negthickspace[n_0\negthinspace -\negthinspace 1]}(t) ) \geq 0
\end{multline}
Hence define $z(\zeta_{\scriptscriptstyle 
\mathbf{1}_{\omega[n_0]}}^{\scriptscriptstyle 
0,\chi\negthickspace[n_0\negthinspace -\negthinspace 1]}(t)) \triangleq \big|  
h_{\bar{r},\sigma} ( \zeta_{\scriptscriptstyle 
\mathbf{1}_{\omega[n_0]}}^{\scriptscriptstyle 
0,\chi\negthickspace[n_0\negthinspace -\negthinspace 1]}(t) ) - 
h_{\bar{r},\sigma}( \chi[n_0-1] ) \big|$.

By the Gr\"onwall inequality, we have further that:
\begin{equation}\label{eq:gronwall_conclusion}
	z(\zeta_{\scriptscriptstyle \mathbf{1}_{\omega[n_0]}}^{\scriptscriptstyle 0,\chi\negthickspace[n_0\negthinspace -\negthinspace 1]}(t))
	\leq
	\sqrt{2} \cdot L_{h_{\bar{r},\sigma}} \cdot \lVert f_{\scriptscriptstyle\text{KBM}}(\chi[n_0-1],\omega[n_0]) \rVert_2 \cdot t \cdot e^{L_{f_\text{KBM}} \cdot t}
\end{equation}
where $L_{h_{\bar{r},\sigma}}$ and $L_{f_{\scriptscriptstyle\text{KBM}}}$ are 
as in the statement of the Lemma. Observe that the function on the right-hand 
side of \eqref{eq:gronwall_conclusion} is monotonic in $t$. Thus, we can use 
\eqref{eq:triangle_bound_pre_gronwall} to claim that if $\nu$ solves 
\eqref{eq:gronwall_conclusion_statement} (derived immediately from 
\eqref{eq:gronwall_conclusion} and \eqref{eq:triangle_bound_pre_gronwall}),  
then the claim of the Lemma holds.
\end{proof}

\begin{table*}[ht]
    \centering
    \caption{EnergyShield Performance across 4 different RL controllers. Each RL agent learnt to travel the route through a distinctive policy represented by its center deviance (CD) from the primary track. The RL Controllers are numerically arranged in the increasing order of CD with Controller 1 being the main RL controller used in all evaluations.}
    \vspace{-2ex}
    \begin{tabular}{c c | c c c | c c c | c c c | c c c }
    \hline
    \multirow{2}{*}{Policy} & \multirow{2}{*}{(S, N)} & \multicolumn{3}{c}{Controller 1} & \multicolumn{3}{c}{Controller 2} & \multicolumn{3}{c}{Controller 3} & \multicolumn{3}{c}{Controller 4} \\
    \cline{3-6} \cline{7-10} \cline{11-14}
    & & CD(m) & TCR(\%) & E(mJ) & CD (m) & TCR(\%) & E(mJ) &  CD(m) & TCR(\%) & E(mJ) & CD(m) & TCR(\%) & E(mJ) \\
    \hline
    \multirow{4}{*}{Local} & (0, 0) & 0.92 & 65.7 & \multirow{4}{*}{113.5} & 2.3 & 100 & \multirow{4}{*}{113.5} & 5.5 & \multirow{4}{*}{100} & \multirow{4}{*}{113.5} & 5.8 & \multirow{4}{*}{100} & \multirow{4}{*}{113.5} \\
    & (0, 1) & 0.8 & 22.9 & & 2.3 & 97.1 & & 5.5 & & & 5.8 & & \\
    & (1, 0) & 2.87 & 100 & & 3.5 & 100 & & 5.5 & & & 5.9 & & \\
    & (1, 1) & 2.91 & 100 & & 3.6 & 100 & & 5.4 & & & 5.7 & & \\
    \hline
    \multirow{3}{*}{EnergyShield} & (0, 0) & 0.8 & 68.6 & 90.8 & 2.3 & 100 & 90 & \multirow{4}{*}{5.5} & \multirow{4}{*}{100} & 79.7 & 5.8 & \multirow{4}{*}{100} & 77.6 \\
    \multirow{3}{*}{(\textit{eager})}& (0, 1) & 0.8 & 34.3 & 90 & 2.3 & 100 & 89 & & & 80 & 5.7 & & 78.9 \\
    & (1, 0) & 2.8 & 100 & 85.8 & 3.5 & 100 & 81.6 & & & 79.6 & 5.7 & & 77.5 \\
    & (1, 1) & 2.8 & 100 & 85.9 & 3.6 & 100 & 81.9 & & & 78.5 & 5.8 & & 77.9 \\
    \hline
    \multirow{3}{*}{EnergyShield} & (0, 0) & 0.9 & 74.3 & 67.7 & 2.3 & 100 & 63.5 & 5.5 & \multirow{4}{*}{100} & 43.7 & 5.8 & \multirow{4}{*}{100} & 39.7\\
    \multirow{3}{*}{(\textit{uniform})}& (0, 1) & 0.7 & 22.9 & 60.6 & 2.3 & 97.1 & 63.1 & 5.5 & & 44.4 & 5.7 & & 40.8 \\
    & (1, 0) & 2.9 & 100 & 51.5 & 3.5 & 100 & 44.5 & 5.5 & & 43.6 & 5.8 & & 40.1 \\
    & (1, 1) & 2.8 & 100 & 53.1 & 3.6 & 100 & 45.7 & 5.4 & & 42.1 & 5.7 & & 39.8 \\   
    \hline

    \end{tabular}
    \label{tab:agents}
    \vspace{-2.5ex}
\end{table*}

\section{Additional Experimental Details}
\label{sec:experimental_details}

\subsection{Training Details} The primary RL agent training was conducted under the (S = 0, N = 0) configuration settings using the Proximal Policy Optimization (PPO) algorithm for a total of ~1800 episodes. In the last 400 training episodes, we randomized the ego vehicle's spawning position and orientation along its lateral dimension to aid the agent in learning how to recover from maneuvering moves. For the $\beta$-VAE, we used the pretrained model from \cite{ferlez2020shieldnn} which was trained to generate a 64-dimensional latent feature vector from Carla driving scenes.

\subsection{Reward Function}For the reward function $\mathcal{R}$, we defined:
\begin{equation}
    \mathcal{R} = 
    \begin{cases} 
    -P, \;\;\;\;\;\;\;\;\;\;\;\;\;\;\;\;\;\;\;\;\;\;\;\;\;\;\;\; \text{collision or }CD > CD_{th} \\
    +P, \;\;\;\;\;\;\;\;\;\;\;\;\;\;\;\;\;\;\;\;\;  \text{track completed successfully}\\ 
    f_{\mathcal{R}}(v, CD, \vartheta, r), \;\;\;\;\;\;\;\;\;\;\;\;\;\;\;\;\;\;\;\;\;\;\;\;\;\;\;\;\; \text{otherwise} \\
    \end{cases} \label{eqn:reward}
\end{equation}
in which $P$ is large positive number, $v$ is the vehicle's velocity, $CD$ is the vehicle's center deviance from the center of the track, $CD_{th}$ is a predetermined threshold value, $\vartheta$ represents the angle between the heading of the vehicle and the tangent to the curvature of the road segment, and $r$ is the distance to the closest obstacle. As shown, $\mathcal{R}$ can evaluate to: (\emph{i}) ($+P$) if it completes the track successfully (large positive reward), (\emph{ii}) ($-P$) if it incurs a collision or deviates from the center of the road beyond $CD_{th}$, or (\emph{iii}) a function $f_{\mathcal{R}}(\cdot)$ of the aforementioned variables given by:
\begin{equation*}
	f_{\mathcal{R}}(v,CD, \vartheta, r) = \omega_1 \cdot f_1(v) + \omega_2 \cdot f_2(\mathcal{CD}) + \omega_3 \cdot f_3(\vartheta) + \omega_4 \cdot f_4(r)
\end{equation*}
 
\begin{align*}
\text{s.t., \;\;\;\;}f_1(v) &=
\begin{cases}
\frac{v}{v_{min}}, \;\;\;\;\;\;\;\;\;\;\;\;\;\;\;\;\;\;\;\;\;\;\;         v < v_{min} \\
1 - \frac{v-v_{target}}{v_{max}-v_{target}}\;\;\;\;\;\;\;\;  v > v_{target} \\     
1,  \;\;\;\;\;\;\;\;\;\;\;\;\;\;\;\;\;\;\;\;\;\;\;\;\;\;\;\; otherwise      
\end{cases} \\
f_2(CD) &= \max(1-\frac{l_{center}}{l_{max}}, 0) \\
f_3(\vartheta) &= \max(1-|\frac{\vartheta}{\pi/9}|, 0) \\
f_4(dist) &= \max(\min(\frac{||r||}{r_{max}}, 1), 0) 
\end{align*}
in which $v_{min}$, $v_{max}$, $v_{target}$ are the minimum, maximum, and target velocities, respectively. $l_{center}$ is the lateral distance from the center of the vehicle to the designated track. $\vartheta$ is the angle between the head of the vehicle and the track's tangent. For our experiments, we set $v_{min}$=35 km/hr, $v_{target}$=40 km/hr, $v_{max}$=45 km/hr, $r_{max}$ = 20 m, $l_{max}$ = 10 m, and $P$= 100.

\subsection{Performance Evaluations}We use the standard TensorRT library to compile our models on the Nvidia Drive PX2 ADS platform as an optimized inference engine and measure its execution latency. To evaluate the local execution power, we measure the difference in average power drawn by the Nvidia Drive PX2 when processing and idling. 

\subsection{Queuing Delays}We leverage the The M/M/1/k model for the queuing delays, $L_{que}$, which entails $q_c = \frac{(1-\rho)(\rho)^c}{1-\rho^{C+1}}$ representing the probability that an offloaded task will find $c$ tasks stored in the server's buffer of size $C$ upon arrival with $\rho$ being the average server load. We assume each task contributes an extra 1 ms delay, and thus, $q_c$ positions directly translate to $L_{que}$ in ms. The default settings for queuing delays entail C = 4000 and $\rho=0.97$ unless otherwise was stated.

\subsection{Edge Response Estimation}As offloading decisions are made based on estimates of the prior edge response times, the estimated communication latency, $\hat{L}_{comm}$, at time $n$ can be defined as a function of the $k$ previous values of the effective throughput $\phi$ and queuing delays $q$ as follows: $\hat{L}_{comm}(n) = \Phi(\phi_{n-k:n-1}, q_{n-k: n-1})$. For our experiments, we set $k=5$ and employ a moving average function to evaluate $\Phi$.

\end{document}